\documentclass[11pt]{article}
\usepackage[margin=1in]{geometry}
\usepackage[utf8]{inputenc}
\usepackage[T1]{fontenc}
\usepackage{lmodern,amsmath,amssymb,amsthm,mathtools,graphicx,booktabs,enumitem,setspace,microtype,float}
\usepackage{needspace}
\usepackage[authoryear,round]{natbib}
\usepackage[colorlinks=true,linkcolor=blue,citecolor=blue,urlcolor=blue]{hyperref}
\newtheorem{assumption}{Assumption}
\newtheorem{definition}{Definition}
\newtheorem{lemma}{Lemma}
\newtheorem{theorem}{Theorem}
\newtheorem{proposition}{Proposition}
\newtheorem{corollary}{Corollary}
\newcommand{\R}{\mathbb R}
\newcommand{\E}{\mathbb E}
\newcommand{\HH}{\mathsf H}
\newcommand{\LL}{\mathsf L}
\newcommand{\Z}{\mathcal Z}
\newcommand{\DD}{\overline D}
\newcommand{\Cmin}{\underline C}
\newcommand{\Cmax}{\overline C}
\newcommand{\ind}{\mathbf 1}

\newcommand{\Cov}{\operatorname{Cov}}

\title{Suspicious Perfection:\\Rating Rewards, Credibility Inversion, and Information Loss}
\author{Van-Quy Nguyen\thanks{Email: \url{quynv@neu.edu.vn}}\\[0.5em]
\small\parbox{0.9\textwidth}{\centering \it Faculty of Fundamental Sciences, College of Technology, National Economics University, Hanoi, Vietnam.}}
\date{October 13, 2026}
\hypersetup{
  pdftitle={Suspicious Perfection: Rating Rewards, Credibility Inversion, and Information Loss},
  pdfauthor={Van-Quy Nguyen},
pdfkeywords={Online ratings; review manipulation; credibility inversion; rating thresholds; Blackwell informativeness}
}

\begin{document}
\maketitle
\begin{abstract}
\noindent Platforms use rating bonuses and badges to reward good sellers. When ratings can be inflated, stronger rewards can instead weaken the information available to buyers. We study an initially informative rating system in which sellers can pay to raise their displayed score and buyers rationally account for manipulation. Manipulation costs and returns are common across qualities conditional on the honest rating. Under this restriction, credibility can fail to increase between adjacent scores only where the higher score earns an additional reward. Our main result concerns large top-rating bonuses. Once sellers choose only between keeping their honest rating and reaching the top, a larger bonus makes the equilibrium rating system less informative in Blackwell's order along either extremal equilibrium selection. Lower scores retain their honest quality composition but become rarer. As the bonus grows, almost everyone displays the top score, its credibility approaches the population average, and the value of rating information vanishes. Nevertheless, a nearly perfect score can remain more credible than a perfect one. We derive the rate of information loss for exponential manipulation costs and establish a budget-feasible exposure rule that avoids credibility reversals. The results suggest that platform evaluations should examine rating frequencies together with category credibility. Rational buyer skepticism can correct the interpretation of ratings without recovering the distinctions that manipulation has erased.
\end{abstract}

\newpage

\section{Introduction}
Online ratings help buyers distinguish sellers they have not yet encountered \citep{dellarocas2003,tadelis2016}. Platforms also use those ratings to award badges, allocate visibility, and decide which sellers receive additional advantages \citep{tadelis2016,hui2023}, and a seller's position in a ranking strongly affects whether consumers search its offer \citep{ursu2018}. A rating therefore performs two jobs: it provides information and determines rewards. These jobs can come into conflict when sellers can purchase favorable reviews. A stronger reward for a high score may encourage sellers to obtain the label rather than improve the quality that the label is meant to represent.

This conflict matters for the governance of digital markets. Platforms organize both the information buyers receive and the incentives sellers face. Research on digital-market policy stresses that these markets tend to be concentrated and that consumer attention and data are key inputs \citep{calvano2021}. We focus on a specific information problem within this broader setting: can a reward intended to recognize good sellers make the rating system less useful even when buyers fully understand manipulation?

Consider a bonus available only to sellers displaying five stars. As the bonus grows, more sellers find it worthwhile to move to five. The remaining sellers displaying $4.9$ can have a better quality composition than the larger group displaying $5.0$. Yet the $4.9$ label helps few buyers if almost nobody displays it. A lower label can remain favorable news while the system as a whole delivers little news. This distinction between \emph{the credibility of a label} and \emph{the information supplied by the system} is the paper's central point.

Evidence motivates the incentives, without establishing the mechanism. Comparing hotel reviews across platforms, \citet{mayzlin2014} find patterns consistent with more promotional reviewing by independent hotels with single-unit owners, whose incentives to fake are strongest, and with more negative reviewing of their nearby competitors. Using Yelp-filtered reviews as a proxy for fraud, \citet{lucazervas2016} find that restaurants with weak reputations are more likely to post fake positive reviews and that those facing more competition are more likely to receive fake negative reviews. Studying purchased Amazon reviews, \citet{he2022} document temporary rating gains and a causal sales effect; subsequent review patterns indicate that manipulation in their sample is mostly used by low-quality products. \citet{maslowska2017}, also summarized by \citet{spiegel2017}, find that purchase probabilities can peak below a perfect rating. This pattern alone does not show that buyers distrust perfection or identify the mechanism studied here: prices, placement, and review counts also affect purchases. We ask the narrower theoretical question of when a higher displayed rating implies a lower probability of high quality, and how this reversal relates to the decision value of ratings.

The model has high- and low-quality sellers. Each starts with an honest rating; higher honest ratings provide stronger evidence of high quality. Sellers can pay to raise their displayed rating, and some sellers find this cheaper than others. Buyers understand these incentives and update their quality beliefs from the displayed score. Seller returns have two components. One increases with buyers' assessment of quality. The other is a reward attached to the number itself, such as an extra placement advantage or a badge-related benefit. A buyer can discount a crowded five-star category without removing the separate reward that attracts sellers to it.

The key restriction is \emph{quality neutrality conditional on the honest rating}. Among sellers starting at the same honest rating, the cost distribution and the returns to each target are common across qualities. Both qualities consequently move in the same proportions. For example, if one-quarter of an honest-$4.7$ category moves upward, one-quarter of its high-quality sellers and one-quarter of its low-quality sellers do so. Departures preserve that category's quality composition; arrivals from lower honest ratings dilute it. The two qualities can still manipulate at different rates overall because their honest starting ratings differ. The restriction isolates the information lost by relabeling honest origins, rather than information potentially revealed by quality-dependent gaming ability.

Our main result concerns an additional reward available only at the top. Above a finite sufficient bonus, sellers choose either to keep their honest rating or to move directly to the top. Intermediate lower categories receive no entrants. Their honest quality composition is therefore preserved, even though their populations shrink. We show that a still larger bonus makes the equilibrium rating system less informative in Blackwell's order along either the smallest or the largest equilibrium solution (Theorem~\ref{thm:information}). The new system can be obtained from the old one by randomly replacing some lower labels with the top label, using a rule that does not separately observe quality. It is consequently weakly less useful for every decision problem.

The result allows both seller behavior and buyer beliefs to adjust. It does not assume uniqueness at every bonus: with multiple equilibria, smallest solutions are compared with smallest solutions, and largest with largest. We provide sufficient conditions for eventual uniqueness, so this qualification disappears at large bonuses for standard cost distributions. The limiting result is stronger in its coverage. Across all equilibria, as the top bonus grows without bound, almost everyone displays the top rating, its credibility approaches the population share of high-quality sellers, and the value of information converges to zero. Meanwhile, the second-highest label retains its honest meaning. If that meaning indicates above-average quality, a positive credibility gap survives even as almost all useful distinctions disappear. For exponential manipulation costs, we obtain the leading rate of this information loss (Corollary~\ref{cor:rate}).

Two supporting results identify the role of rewards. First, credibility must strictly increase between adjacent scores whenever the reward does not increase (Theorem~\ref{thm:location}). A higher target with neither more reward nor more credibility would be dominated by the cheaper adjacent target. It would receive no entrants and regain its more favorable honest composition, contradicting the proposed reversal. Second, a sufficiently large reward step removes inflow into lower categories and attracts sellers to the threshold. Under an explicit sufficient condition on the honest distribution and the cost of moving past the threshold, the threshold becomes less credible than the category immediately below it (Theorem~\ref{thm:steps}). At the top, it suffices that the honest second-highest rating indicate above-average quality.

The policy implication is a distinction between correcting beliefs and preserving information. Rational buyers can correctly discount inflated ratings, but they cannot recover the honest origins hidden inside a pooled label. An evaluation of a rating policy should therefore examine the frequency of labels together with their quality composition. We also give a credibility-linked exposure rule that respects a fixed budget and rules out reversals (Proposition~\ref{prop:allocation}). This establishes feasibility and monotone credibility, rather than optimality or a welfare improvement. A full welfare analysis would require the costs of rewards and manipulation, the allocation of buyer attention, and possible incentives to improve quality.

\subsection{Related literature and contribution}
Reputation research explains how feedback supports information and trade \citep{dellarocas2003,tadelis2016}. Manipulation with rational consumers is already well established. \citet{dellarocas2006} relates its effects to how manipulation varies with firm quality; \citet{akoz2020} study review manipulation jointly with pricing; \citet{mayzlin2006} studies disguised promotion; and \citet{glazer2021} show that fake reviews slow learning while valuable learning continues in every period, and that filtering or garbling reported reviews cannot improve receiver welfare within their dynamic model. We hold prices and quality fixed and impose common transitions conditional on the honest rating. This permits an explicit comparison of the information systems induced by different platform rewards. Our contribution is the link between reward-induced movements across categories, the surviving meaning of rare labels, and the decision value of the entire rating system.

The closest theoretical connection is to manipulable signals. \citet{frankel2019} show that higher signaling stakes can reduce equilibrium information about natural actions in Blackwell's order in a $2\times2$ environment; under additional conditions, high stakes also make posterior means approximately uninformative. \citet{frankel2022} study commitment to weaker responses to manipulable data, while \citet{perez2022} analyze test design under falsification. We focus on an additional reward attached to a displayed rating, separate from the return to buyers' posterior beliefs. On a finite rating grid, this permits us to locate credibility reversals, construct an explicit garbling along extremal equilibrium selections above a sufficient top-bonus threshold, and show how rare lower labels retain their honest meaning while information value vanishes. The contribution is the link between persistent credibility inversion and vanishing information value under these rating-specific incentives.

\citet{kartik2009} obtains pooling on the highest messages in a model of costly lying. \citet{chan2007} study grade inflation by schools unable to commit to honest grading. \citet{feltovich2002} show how additional noisy information can support countersignaling, making a less conspicuous action favorable news. These mechanisms provide precedents for inflated or nonmonotone signals. We do not claim that suspicious perfection or upper-bound pooling is new by itself. The result separates an inversion between two labels from a loss of information across the whole system. Blackwell's comparison of experiments accounts for both the meanings and the probabilities of signals \citep{blackwell1953}. It therefore explains why a persistent credibility reversal can accompany the disappearance of information value.

Displayed-rating thresholds affect demand \citep{anderson2012}, certification requirements affect seller behavior and market composition \citep{hui2023}, and rankings affect which offers consumers inspect \citep{ursu2018}. These distinct mechanisms motivate the model's reward components; they do not directly establish its predicted credibility inversion or information loss. For instance, Airbnb includes an overall rating of at least $4.8$ among several requirements for Superhost status \citep{airbnb2026}. This is an example of a rating criterion, not evidence that the particular policy generates the inversion studied here.

\subsection{Scope and organization}
The analysis isolates the informational effect of an exogenous reward for a displayed rating. Quality and prices are fixed, manipulation is upward, and buyers observe a coarse score without review counts or history. Conditional quality neutrality is essential to the composition and garbling arguments. If costs or returns differ by quality within an honest category, the choice to manipulate can itself reveal quality and the information comparison need not survive. Our numerical parameters illustrate the mechanism; they are not estimates for a specific platform.

Section~\ref{sec:model} presents the model. Section~\ref{sec:credibility} studies category composition and credibility reversals. Section~\ref{sec:information} establishes the information comparison, and Section~\ref{sec:example} illustrates it numerically. Section~\ref{sec:policy} discusses platform evaluation, exposure design, and implementation. Appendix~\ref{app:foundations} provides the equilibrium foundations; Appendix~\ref{app:proofs} proves the main results; and Appendix~\ref{app:technology} derives target costs from review production.

\section{Model}\label{sec:model}
\subsection{Honest ratings and manipulation costs}
We normalize the seller population to one, so the mass of a group is its share of all sellers. Each seller has fixed, privately known quality $\theta\in\{\HH,\LL\}$, either high or low. The population share of high-quality sellers is $\Pr(\HH)=\pi\in(0,1)$. Before seeing a rating, buyers therefore assign probability $\pi$ to high quality; this is their \emph{prior}. Displayed ratings lie on a finite grid,
\[
 \Z=\{z_1<\cdots<z_K\},\qquad K\geq2.
\]
The grid represents the rounded averages or categories that buyers actually see.

A seller's \emph{honest rating} $a\in\Z$ is the rating it would display without manipulation. Write $h_\theta(a)>0$ for the probability of this rating conditional on quality $\theta$, with $\sum_a h_\theta(a)=1$. Three quantities describe the honest data:
\[
 \ell^0(a)=\frac{h_\HH(a)}{h_\LL(a)},\quad
 H_\theta(z)=\sum_{a\leq z}h_\theta(a),\quad
 f(a)=\pi h_\HH(a)+(1-\pi)h_\LL(a).
\]
The likelihood ratio $\ell^0(a)$ compares how often rating $a$ occurs among the two qualities. A value of two, for example, means that this rating is twice as likely among high-quality sellers as among low-quality sellers. It is evidence about quality, not itself a probability. The cumulative probability $H_\theta(z)$ is the share of type-$\theta$ sellers whose honest rating is at most $z$. Finally, $f(a)$ is the share of all sellers with honest rating $a$. These quantities describe, respectively, the information in a label, the population below a threshold, and the frequency of a label.

\begin{assumption}[Honest information]\label{ass:honest}
The honest likelihood ratio $\ell^0$ is strictly increasing on $\Z$.
\end{assumption}

The assumption says that a higher honest rating provides stronger evidence of high quality. Bayes' rule converts that evidence into the \emph{honest posterior}, the probability of high quality after observing an unmanipulated rating:
\begin{equation}\label{eq:honestposterior}
 \mu^0(a)=\frac{\pi h_\HH(a)}{f(a)}.
\end{equation}
The numerator is the population mass of high-quality sellers with rating $a$; the denominator is the mass of all sellers with that rating. Thus $\mu^0(a)$ is simply the high-quality share within the category. It is strictly increasing in $a$. Any later reversal must therefore come from manipulation rather than from an initially misleading scale.

A seller can raise its displayed rating from $a$ to any $z\geq a$ at cost $\kappa C(a,z)$. The finite, nonnegative function $C$ measures the difficulty of reaching a target. The private parameter $\kappa>0$ scales that difficulty: a seller with a smaller $\kappa$ can manipulate more cheaply. It is independent of quality and the honest rating and has distribution function $G$.

\begin{assumption}[Manipulation technology]\label{ass:cost}
\begin{enumerate}[label=(\roman*)]
\item $C(a,a)=0$, and $C(a,z')>C(a,z)$ whenever $z'>z\geq a$.
\item $G$ is continuous, $G(x)=0$ for $x\leq0$, and $0<G(x)<1$ for every $x>0$. Write $S=1-G$ and set $G(+\infty)=1$.
\end{enumerate}
\end{assumption}

Part (i) says that keeping the honest rating is free and reaching a higher target is more expensive. Part (ii) rules out a positive mass of sellers at any single cost parameter. It also ensures that there are sellers with very low costs and sellers with arbitrarily high costs. Consequently, sufficiently strong incentives can induce some sellers to manipulate, but no finite reward induces everyone to do so. The survival function $S(x)=1-G(x)$ is the fraction with costs above $x$.

A displayed five must be distinguished from an exact average of five. If the platform rounds averages, a displayed five can be reached once the average crosses the lower edge of its rounding bin. An exact average of five cannot be reached from a lower starting average by adding finitely many five-star reviews. Appendix~\ref{app:technology} derives the target cost $C(a,z)$ from this review arithmetic. The main analysis uses the minimum cost of reaching a displayed category, so it need not track every fake review separately.

The common cost distribution is a substantive restriction. Among sellers with the same honest rating, both qualities face the same distribution of manipulation costs and the same returns. Their departure rates are therefore equal. For example, if one-quarter of sellers with an honest $4.7$ move to a higher category, one-quarter of its high-quality sellers and one-quarter of its low-quality sellers do so. The high-quality share among those who remain at $4.7$ is unchanged. Because departures preserve the quality composition, a category recovers its honest meaning once inflow from below stops.

We call this restriction \emph{quality neutrality conditional on the honest rating}. It allows the two qualities to manipulate at different rates overall: high-quality sellers are more likely to start from better honest ratings, and the costs of reaching a target can differ across starting points. What it excludes is an additional quality-dependent advantage in manipulation among sellers who start at the same rating. The composition argument could allow the distribution $G$ to vary with the honest rating $a$, provided it remained common across qualities within that rating; we retain one $G$ throughout.

The population has already formed its honest ratings when the platform reward is evaluated. The model does not treat manipulation as investment in actual quality. The action set permits only upward manipulation. Adding negative reviews, deleting favorable reviews, or resetting an account is outside the model. Buyers observe the displayed rating alone, without review counts or the complete review history. The sequence is therefore: an announced reward rule; private observation of quality, the honest rating, and the cost parameter; a choice of displayed rating; and buyer inference and demand. There is no learning over repeated transactions in this model.

\subsection{Demand, rewards, and equilibrium}
Let $\mu(z)$ be buyers' equilibrium probability of high quality after seeing displayed rating $z$. We call this probability the rating's \emph{credibility}. A seller receives a return $\rho D(\mu(z))$ from informed buyers, where $\rho>0$ and $D:[0,1]\to\R$ is continuous and strictly increasing. The parameter $\rho$ measures the importance of this demand component. Define $\DD=D(1)-D(0)>0$, so $\rho\DD$ bounds the possible gain from improving buyers' beliefs.\footnote{Every equilibrium posterior lies between $\mu^0(z_1)$ and $\mu^0(z_K)$ (Lemma~\ref{lem:composition}). At equilibrium, demand differences are therefore bounded by the tighter range $D(\mu^0(z_K))-D(\mu^0(z_1))$. We retain $\DD$ throughout because it also bounds differences at candidate beliefs used in the existence proofs and over the scalar enclosure used in the contraction argument.} Prices and production margins are fixed, and the return from a current sale does not depend on the seller's quality.

In addition, the seller receives a \emph{mechanical reward} $r(z)$ attached directly to the displayed number. This is a payoff-equivalent reward, which can represent a placement advantage, a certification benefit, or demand that responds directly to the stars. The term ``mechanical'' means that the rule pays for the label without adjusting to its equilibrium credibility. It does not mean that informed buyers misunderstand ratings. A seller may continue to earn the reward at five stars even when those buyers regard five stars with suspicion.

The schedule $r$ is announced before manipulation and is common knowledge: the platform sets the rule, and we solve for the seller behavior and buyer beliefs it induces. This assumption lets us evaluate an announced policy without solving the platform's full design problem. Section~\ref{subsec:allocation} considers a budget-feasible alternative.

The seller chooses a feasible rating to maximize
\begin{equation}\label{eq:payoff}
 \Pi(z;a,\kappa)=r(z)+\rho D(\mu(z))-\kappa C(a,z),\qquad z\geq a.
\end{equation}
Equation~\eqref{eq:payoff} says that a seller balances two benefits against one cost: the reward for reaching the number, the demand generated by its credibility, and the expense of producing it. A higher label can be attractive because of the first benefit even when the second is smaller. Write $V(z)=r(z)+\rho D(\mu(z))$ for the gross return before manipulation costs. Quality affects the distribution of honest starting ratings, but does not otherwise enter this choice problem.

For a strategy $\zeta(a,\kappa)\geq a$, let $p(a,z)$ be the fraction of sellers with honest rating $a$ who display $z$. These fractions integrate over the private cost parameter. The displayed-rating distributions conditional on quality, and the total mass at each displayed rating, are
\begin{equation}\label{eq:distributions}
 \phi_\theta(z)=\sum_{a\leq z}h_\theta(a)p(a,z),\qquad
 m(z)=\pi\phi_\HH(z)+(1-\pi)\phi_\LL(z).
\end{equation}
The sum collects all sellers who reach $z$, whether they start there or arrive from below. Buyers' beliefs must agree with the composition generated by those choices:
\begin{equation}\label{eq:bayes}
 \mu(z)=\frac{\pi\phi_\HH(z)}{m(z)},\qquad
 \ell(z)=\frac{\phi_\HH(z)}{\phi_\LL(z)}.
\end{equation}
As in the honest benchmark, credibility is the high-quality share among sellers displaying the rating. The likelihood ratio $\ell(z)$ expresses the same information in relative-probability form. There are sufficiently many sellers that one seller alone does not change the composition of a rating category. It therefore takes rewards and buyer beliefs as fixed when choosing its rating.

\begin{definition}\label{def:equilibrium}
An equilibrium is a measurable strategy $\zeta$ and a belief vector $\mu$ such that $\zeta(a,\kappa)$ maximizes \eqref{eq:payoff} for every $(a,\kappa)$, every displayed rating has positive mass under both qualities, and beliefs obey \eqref{eq:bayes}.
\end{definition}

Equilibrium thus requires consistency in both directions: sellers choose ratings in response to buyers' beliefs, and buyers form those beliefs from sellers' choices. Ties occur with zero probability under the continuous cost distribution and therefore do not affect category composition. We say that a rating $z$ \emph{receives inflow} when $p(a,z)>0$ for some $a<z$. Its \emph{honest residents} are sellers who start at $z$ and remain there. Uniqueness, when established, disregards differences in behavior on a set of cost parameters with probability zero.

\section{How rewards change rating credibility}\label{sec:credibility}
\subsection{The composition of a displayed category}
An equilibrium exists under the stated assumptions, and every label retains a positive mass of honest sellers of each quality (Proposition~\ref{prop:existence}). Some sellers always have costs high enough to make manipulation unattractive. This guarantees that buyers can apply Bayes' rule at every displayed label. Appendix~\ref{app:foundations} derives the exact choice intervals and proves existence. The main text begins with their economic implication: how the origins of sellers determine what a label means.

\begin{lemma}[Composition and garbling]\label{lem:composition}
In every equilibrium,
\begin{equation}\label{eq:composition}
 \mu(z)=\sum_{a\leq z}w_z(a)\mu^0(a),\qquad
 w_z(a)=\frac{f(a)p(a,z)}{m(z)},
\end{equation}
where the weights sum to one and $w_z(z)>0$. Hence $\mu(z)\leq\mu^0(z)$, with equality exactly when $z$ receives no inflow. In particular, $\mu(z_1)=\mu^0(z_1)$. The displayed-rating experiment is a Blackwell garbling of the honest-rating experiment.
\end{lemma}

Equation~\eqref{eq:composition} has a direct composition interpretation. The weight $w_z(a)$ is the share of sellers at displayed rating $z$ who came from honest origin $a$. The posterior at $z$ is the weighted average of the honest quality shares of these origins. Suppose half the sellers at a label came from an origin with a $0.2$ high-quality probability and half came from an origin with probability $0.8$. The label's credibility would be $0.5$. The formula makes exactly this calculation for all possible origins. Appendix~\ref{app:proofs} also gives the equivalent likelihood-ratio representation used in the threshold comparison.

Two features determine the direction of the effect. Entrants arrive only from lower honest ratings, which carry less favorable information. Departures remove high- and low-quality residents in the same proportions. Inflow therefore reduces credibility, while departures alone leave the category's credibility unchanged. In particular, the lowest rating keeps its honest meaning because nobody can enter it from below.

The last statement concerns the whole rating system. An \emph{experiment} is a signal together with the probabilities of observing it under each possible quality. A \emph{Blackwell garbling} is a way to make that signal less precise by relabeling it randomly, without separately observing quality. Here the relabeling rule is $p(a,z)$: start with an honest rating and apply the transition generated by manipulation. A buyer with access to the honest rating could reproduce the displayed signal by applying this rule, or could use the original information directly. The honest rating is therefore at least as useful for every decision problem.

This ranking requires a common transition across qualities. If high- and low-quality sellers faced different costs within an honest category, the decision to manipulate could itself reveal quality. The displayed signal would then combine honest information with this additional selection, and the quality-independent garbling argument would no longer apply.

\subsection{Where credibility can reverse}
A credibility inversion occurs when a higher displayed rating has a lower posterior probability of high quality. The next theorem gives a necessary condition even for a weak reversal, where the two probabilities are equal.

\begin{theorem}[Location of inversions]\label{thm:location}
In every equilibrium and for every $k\geq2$,
\begin{equation}\label{eq:location}
 r(z_k)\leq r(z_{k-1})\quad\Longrightarrow\quad
 \mu(z_k)>\mu(z_{k-1}).
\end{equation}
Consequently, credibility is strictly increasing on every consecutive block of ratings with constant rewards. In particular, a constant reward schedule rules out inversion everywhere.
\end{theorem}

Suppose the higher rating offered neither a larger mechanical reward nor greater credibility. Any seller able to manipulate into it could instead reach the adjacent lower rating at strictly lower cost. Nobody would then enter the higher category. By Lemma~\ref{lem:composition}, that category would retain its honest posterior, while the category below could be at most as credible as its own honest benchmark. Strictly increasing honest posteriors would make the higher rating more credible after all. This contradiction proves the result.

Under the maintained neutrality restriction, cheap fake reviews alone do not explain suspicious perfection. If rewards depend only on credibility, a perfect rating remains better news than the rating immediately below it. A rise in the mechanical reward makes an inversion possible, but does not guarantee one.

\begin{corollary}[The ends of the scale]\label{cor:ends}
In every equilibrium: (i) the lowest rating retains its honest posterior; (ii) $\mu(z_2)>\mu(z_1)$ for every reward schedule; and (iii) if $r$ is nondecreasing, the top rating receives inflow and $\mu(z_K)<\mu^0(z_K)$.
\end{corollary}

Part (ii) explains why at least three displayed categories are needed for an inversion. The second-lowest category contains only its own honest residents and possible entrants from the lowest category. Because some honest residents always remain, its credibility stays strictly above the lowest category's. Part (iii) makes a different point: under nondecreasing rewards, manipulation always reduces the credibility of the top rating relative to its honest benchmark. This \emph{dilution} need not be an inversion. A displayed five can be less credible than an honest five would be and still be the most credible rating available.

The maintained restriction matters for the breadth of the result. If only low-quality sellers could manipulate, a similar dominance argument would still rule out a top inversion without a larger top reward. However, selective departures could improve an interior category, so the general composition argument would need to change. The results below maintain common transitions within each honest origin.

\subsection{Reward steps, pooling, and inversion}\label{subsec:steps}
Consider a reward that becomes available at a specified threshold. For $j\in\{2,\ldots,K\}$, let
\begin{equation}\label{eq:step}
 r_\Delta(z)=r_0+\Delta\ind_{\{z\geq z_j\}},\qquad \Delta\geq0.
\end{equation}
The indicator equals one when the displayed rating reaches $z_j$ and zero otherwise. Thus $\Delta$ is the extra reward for qualifying. If $j<K$, all ratings at or above the threshold receive it; if $j=K$, only the top rating does.

To separate the incentive to reach the threshold from the incentive to go beyond it, put $\Cmax=\max_{a\leq z}C(a,z)$ and define, for $j<K$,
\[
 c_j=\min_{a\leq z_j}\{C(a,z_{j+1})-C(a,z_j)\}>0,
 \qquad \beta_j=\frac{\rho\DD}{c_j};\qquad \beta_K=0.
\]
Here $c_j$ is the smallest additional base cost of moving one step beyond the threshold. There is no further mechanical reward for doing so, and the credibility benefit cannot exceed $\rho\DD$. Hence a seller whose honest rating is at most $z_j$ and whose cost parameter exceeds $\beta_j$ never chooses a rating above the threshold. The fraction $1-G(\beta_j)$ will provide a conservative lower bound on limiting retention at the threshold.

Let $\Lambda_j$ denote the finite sufficient bonus bound defined explicitly in \eqref{eq:cleanbound} in Appendix~\ref{app:proofs}. Its role is to make every unrewarded intermediate target unattractive: a seller willing to pay for such a target prefers either to stay at its honest rating or to aim at the rewarded threshold. The bound checks every origin--target comparison and works across all equilibria. It need not equal the exact critical bonus in a particular numerical example.

\Needspace{18\baselineskip}
\begin{theorem}[Reward steps]\label{thm:steps}
Under \eqref{eq:step}:
\begin{enumerate}[label=(\roman*)]
\item If $\Delta>\Lambda_j$, no rating below $z_j$ receives inflow, and $\mu(z)=\mu^0(z)$ for all $z<z_j$ in every equilibrium.
\item There is a finite $\Delta_1\geq\Lambda_j$ such that, whenever $\Delta>\Delta_1$, the threshold receives inflow from every lower honest rating in every equilibrium. Thus $\mu(z_j)<\mu^0(z_j)$.
\item If
\begin{equation}\label{eq:interiorcondition}
 \frac{H_\HH(z_j)}{H_\LL(z_j)}
 <[1-G(\beta_j)]\ell^0(z_{j-1}),
\end{equation}
then for some finite $\Delta_2\geq\Delta_1$, every equilibrium at every $\Delta>\Delta_2$ satisfies $\mu(z_j)<\mu(z_{j-1})$.
\end{enumerate}
\end{theorem}

Part (i) says that a large reward restores the honest meaning of ratings below the threshold. A seller who would otherwise pay to reach an intermediate category prefers either to stay at its honest rating or to aim at least as high as the rewarded threshold. Once inflow into a lower category stops, the sellers who remain there have its original quality composition.

The category's recovered meaning says nothing about how many sellers still display it. We use \emph{cleaning} to mean that no sellers manipulate into the category, and \emph{thinning} to mean that its remaining population falls. Cleaning concerns composition; thinning concerns frequency. A reward can produce both at once, leaving a credible label that buyers almost never see. This distinction leads to the information result in Section~\ref{sec:information}. For a top reward, every lower category eventually has vanishing mass; the numerical badge example also shows pronounced thinning below its threshold. Certification requirements can change seller behavior \citep{hui2023}, and displayed-rating thresholds affect demand \citep{anderson2012}. These findings motivate the incentive to reach a threshold; they do not establish the manipulation-induced concentration predicted by this model.

Part (ii) establishes \emph{pooling}: sellers with different honest origins choose the same displayed threshold, which weakens its quality evidence. Part (iii) asks whether this dilution is large enough to reverse the comparison with the preceding category. The left side of \eqref{eq:interiorcondition}, $H_\HH(z_j)/H_\LL(z_j)$, is the likelihood ratio associated with having any honest rating at or below the threshold. It describes the broad group from which the threshold draws sellers, rather than the original threshold residents alone.

At an interior threshold, some sellers can move farther upward. The factor $1-G(\beta_j)$ is a conservative allowance for this onward movement: as the bonus grows, the guaranteed fraction of each origin that ends at the threshold rises toward this factor. Condition~\eqref{eq:interiorcondition} places the resulting upper bound on threshold credibility below the honest credibility of the preceding category. At the top, nobody can move farther, so this allowance is unnecessary.

The condition is sufficient, not necessary. It can be conservative when $\beta_j$ is large, as the numerical illustration will show. Also, dilution alone does not establish inversion: the second-lowest category can never become less credible than the lowest, as Corollary~\ref{cor:ends}(ii) makes clear.

For a top-rating reward, the condition simplifies because there is nowhere further to go: $H_\theta(z_K)=1$ and $\beta_K=0$.

\begin{corollary}[Suspicious perfection]\label{cor:top}
If $\ell^0(z_{K-1})>1$, a sufficiently large top-rating bonus makes $\mu(z_K)<\mu(z_{K-1})$ in every equilibrium. Every lower rating retains its honest posterior once $\Delta>\Lambda_K$.
\end{corollary}

The inequality $\ell^0(z_{K-1})>1$ is equivalent to $\mu^0(z_{K-1})>\pi$. It says that the honest second-highest rating indicates a better-than-average chance of high quality. A large top bonus attracts nearly the whole population to the top, so its credibility approaches the population average. The second-highest rating keeps its above-average honest meaning. Buyers then rationally prefer the information conveyed by a nearly perfect rating to that conveyed by a perfect one.

A step is one way to create this incentive, but it is not necessary. On a finite grid, a reward that rises at every rating can also permit an inversion. The location theorem identifies where inversion is possible; the step theorem supplies conditions under which a particular reward schedule produces it. Neither result says that every increasing reward schedule must reverse credibility.

\section{The decision value of the rating system}\label{sec:information}
\subsection{How much do buyers learn from ratings?}
Credibility describes what a particular rating means when it is observed. To evaluate the rating system, we must also consider how often each rating is observed. A highly informative category contributes little to buyers' decisions if almost no seller displays it.

Let $B(\mu)$ be the value of choosing the best action when the probability of high quality is $\mu$. We take $B:[0,1]\to\R$ to be continuous and convex. To see why convexity is natural, fix an action: its expected payoff is a probability-weighted average of its payoffs under high and low quality, and hence is linear in $\mu$. Taking the best among available actions gives a convex value function: it is the upper envelope of those linear payoff functions. Convexity here expresses the value of being able to adapt a decision to information; it is not an assumption that buyers prefer risk. Better information can help because the buyer can adapt the action to the seller, while retaining the option to ignore the signal. Define
\begin{equation}\label{eq:information}
 I_B=\sum_{z\in\Z}m(z)B(\mu(z))-B(\pi).
\end{equation}
The first term is expected decision value after observing the rating, weighted by the frequency of each label. The second is decision value using only the prior. Thus $I_B$ measures the benefit of access to ratings for a seller drawn at random from the population. Keeping this sampling rule fixed isolates information quality. If a platform also changes which sellers buyers encounter, an exposure-weighted market outcome requires a separate calculation.

Bayes' rule implies $\sum_zm(z)\mu(z)=\pi$: averaging beliefs across all observed ratings returns the population share of high-quality sellers. This accounting identity is called \emph{Bayesian plausibility}. Ratings do not raise average quality or the average belief about it; they help buyers distinguish one seller from another. Convexity makes this distinction valuable, giving $I_B\geq0$. A Blackwell comparison goes beyond checking one choice of $B$. It ranks the signals for every decision problem, and therefore for every continuous convex posterior value function considered here.

For a concrete purchase interpretation, suppose a buyer receives net utility $\ind_{\{\theta=\HH\}}-\xi$ from buying and zero from not buying. The benefit of high quality is normalized to one, and $\xi\in[0,1]$ is the buyer's net purchase cost or friction, with distribution $F$. After observing posterior $\mu$, the buyer purchases if $\mu\geq\xi$. At a fixed encounter and fixed purchase terms, demand is therefore $D(\mu)=F(\mu)$, and expected buyer surplus is
\[
 B_F(\mu)=\E[(\mu-\xi)_+]=\int_0^\mu F(t)\,\mathrm dt.
\]
Here $(x)_+=\max\{x,0\}$ reflects the option not to buy. A buyer with friction $\xi$ obtains surplus $\mu-\xi$ when purchasing is worthwhile and zero otherwise. Averaging the surplus from these optimal purchase decisions over buyers gives $B_F$. The integral adds up demand over possible quality beliefs; equivalently, $B_F'(\mu)=F(\mu)$ wherever the derivative is taken. If $F$ is continuous and strictly increasing, $B_F$ is strictly convex. In this specification, $I_{B_F}$ is expected buyer surplus with ratings minus expected surplus without them.

For uniformly distributed frictions, the expression simplifies to
\begin{equation}\label{eq:variance}
 I_{B_F}=\frac12\sum_zm(z)[\mu(z)-\pi]^2.
\end{equation}
Information value is one-half of the variance of buyers' posteriors. Ratings are useful when they lead buyers to distinguish sellers from the population average, and the squared differences must be weighted by how often buyers encounter them. A large difference at a rare rating need not imply a large information value. This measure is not total welfare: it excludes manipulation costs, the cost of providing rewards, and any benefits from quality investment.

\subsection{Larger top rewards make honest ratings rare}
For the rest of this subsection, the reward in \eqref{eq:step} applies only to the top rating, so $j=K$. Define
\[
 \overline\Delta=\max\{\Lambda_K,\rho\DD\},\qquad C_a=C(a,z_K)\quad(a<z_K).
\]
When $\Delta>\overline\Delta$, Theorem~\ref{thm:steps}(i) removes all inflow into lower intermediate ratings. Each seller then either keeps its honest rating or moves directly to the top. We call this the \emph{large-bonus region}. The bound is sufficient; this simple choice pattern may emerge at smaller bonuses in a particular example.

The entire equilibrium can now be described by one gross-return variable,
$v=\Delta+\rho D(\mu(z_K))$, which omits the common reward $r_0$. For a seller with honest rating $a<z_K$, moving to the top yields an additional return $v-\rho D(\mu^0(a))$ and costs $\kappa C_a$. The seller stays at $a$ when its cost parameter exceeds the ratio of these two quantities. Accordingly, define
\begin{align}
 s_a(v)&=S\left(\frac{v-\rho D(\mu^0(a))}{C_a}\right),\qquad a<z_K,\label{eq:survival}\\
 \varepsilon(v)&=\sum_{a<z_K}f(a)s_a(v),\label{eq:epsilon}\\
 b(v)&=\frac{\pi-\sum_{a<z_K}f(a)s_a(v)\mu^0(a)}{1-\varepsilon(v)}.\label{eq:topbelief}
\end{align}
The equations track retention, frequency, and quality composition in that order. The fraction $s_a(v)$ consists of sellers whose cost parameter exceeds the gain from reaching the top divided by its base cost. Multiplying by the original category mass $f(a)$ and summing gives the total population mass $\varepsilon(v)$ remaining below the top. The top's high-quality mass is then the population total $\pi$ minus the high-quality mass left below it. Dividing by the top's total mass $1-\varepsilon(v)$ gives its credibility $b(v)$. The denominator is at least $f(z_K)>0$, since honest top sellers have nowhere higher to go.

These formulas rely on departures preserving the quality composition of each honest category. Each lower category can shrink while retaining the same posterior $\mu^0(a)$. The changing top category absorbs those departures, so its posterior must be calculated again at every return $v$.

\Needspace{9\baselineskip}
\begin{lemma}[Scalar representation for large top rewards]\label{lem:scalar}
For $\Delta>\overline\Delta$, equilibria correspond exactly to roots of
\begin{equation}\label{eq:scalar}
 v-\rho D(b(v))=\Delta.
\end{equation}
The root set is nonempty and compact. At a root, lower ratings have masses $f(a)s_a(v)$ and posteriors $\mu^0(a)$, while the top has mass $1-\varepsilon(v)$ and posterior $b(v)$. Suppose in addition that $D$ is Lipschitz with constant $L_D$ and that $G$ has a density $g$ continuous on $(0,\infty)$.
\begin{enumerate}[label=(\roman*)]
\item If costs are exponential, $g(x)=\gamma e^{-\gamma x}$, the root is unique whenever
\begin{equation}\label{eq:scalarunique}
 \frac{\rho L_D\gamma}{f(z_K)}\sum_{a<z_K}\frac{f(a)}{C_a}
 \exp\left(-\frac{\gamma(\Delta-\rho\DD)}{C_a}\right)<1.
\end{equation}
\item If $g(x)\to0$ as $x\to\infty$, there is $\Delta_u\geq\overline\Delta$ such that the root, and hence the equilibrium, is unique for every $\Delta>\Delta_u$.
\end{enumerate}
\end{lemma}

Equation~\eqref{eq:scalar} can also be read as $v=\Delta+\rho D(b(v))$. Start with a candidate return $v$. It determines who moves, which determines the top posterior $b(v)$, which in turn determines informed demand. A root is a return that agrees with the demand generated by the choices it induces. This is the original equilibrium feedback reduced to one equation, not an assumption that buyers' beliefs stay fixed when rewards change.

The uniqueness conditions bound the strength of that feedback. A function $D$ is Lipschitz with constant $L_D$ if $|D(u)-D(w)|\leq L_D|u-w|$: demand cannot react arbitrarily strongly to a small belief change. The left side of \eqref{eq:scalarunique} bounds the resulting change in the right side of the fixed-point equation when $v$ changes by one unit. If this bound is below one, feedback is too weak to support multiple roots. More generally, when the cost density becomes small in its upper tail, a large bonus leaves few sellers near the stay-or-move cutoff. Small return changes then move few sellers, and equilibrium is eventually unique.

At intermediate bonuses, several roots may remain possible. Let $v_-(\Delta)$ and $v_+(\Delta)$ denote the smallest and largest roots of \eqref{eq:scalar}. They coincide when equilibrium is unique. These two selections let us make comparisons across bonuses without assuming uniqueness in advance.

\begin{theorem}[Information loss from stronger top rewards]\label{thm:information}
\begin{enumerate}[label=(\roman*)]
\item On $(\overline\Delta,\infty)$, each extremal selection $v_-$ and $v_+$ is strictly increasing in $\Delta$. Along either selection, the mass at every lower rating weakly decreases. For $\Delta_2>\Delta_1>\overline\Delta$, the experiment at $\Delta_2$ is a Blackwell garbling of the experiment at $\Delta_1$ on the same selection. Thus $I_B$ weakly decreases for every continuous convex $B$.
\item If $G$ is strictly increasing on $(0,\infty)$, each lower-rating mass falls strictly along these selections. For every strictly convex $B$, $I_B$ falls strictly between any two such bonuses.
\item Uniformly over all equilibria as $\Delta\to\infty$,
\begin{equation}\label{eq:collapse}
 m(z_K)\longrightarrow1,\qquad \mu(z_K)\longrightarrow\pi,
 \qquad I_B\longrightarrow0,
 \qquad \mu(z_{K-1})-\mu(z_K)\longrightarrow\mu^0(z_{K-1})-\pi .
\end{equation}
Nevertheless $I_B>0$ at every finite bonus for strictly convex $B$, and every lower rating retains its honest posterior once $\Delta>\overline\Delta$. If in addition $C_a$ is strictly decreasing in $a$, then $\mu(z_K)>\pi$ in every equilibrium for all sufficiently large $\Delta$, so the credibility gap converges to its limit from below.
\end{enumerate}
\end{theorem}

Part (i) can be understood by a simple relabeling procedure. Compare two bonuses on the same equilibrium selection, with $v_2>v_1$. The larger return leaves weakly fewer honest-$a$ sellers in every lower category. Starting from the old displayed ratings, keep a lower label $a$ with probability $s_a(v_2)/s_a(v_1)$ and replace it with the top label otherwise. Leave all old top labels unchanged. For example, if an honest origin has retention rates $0.7$ and $0.4$ at the two bonuses, the new system keeps $0.4/0.7$ of its old lower labels and replaces the rest by the top label. The ratio preserves exactly the required fraction of the old category, without asking whether any individual seller is high or low quality.

This procedure reproduces the new rating distribution separately for each quality, using only the old label. In effect, the stronger bonus erases some of the distinctions that the earlier rating system made available. This is why the comparison applies to every decision problem, rather than only to a selected measure of buyer surplus.

This construction explains the loss of decision value. The new top category combines sellers whom the old system could distinguish. Buyers lose some opportunities to condition their actions on those distinctions. When $G$ is strictly increasing, every lower category loses positive mass. The new top then combines old labels with different posteriors, which gives a strict loss for every strictly convex $B$.

The equilibrium qualification is important. With multiple roots, the theorem compares smallest with smallest, or largest with largest. It does not rank arbitrary switches between roots. Once equilibrium is unique, the qualification disappears: increasing the bonus makes the rating system less informative. The limiting conclusions in part (iii) hold across all equilibria, even without uniqueness or a chosen selection.

The limits in \eqref{eq:collapse} describe four aspects of the same outcome. The top mass approaches one because the increasing bonus eventually outweighs the manipulation cost for almost every seller. The top posterior approaches $\pi$ because a category containing almost the whole population must have almost the population's quality composition. Information value approaches zero because buyers almost always see this uninformative label. The credibility gap can remain positive because the second-highest category still has its honest composition, even though it is almost never observed.

The collapse is a limit, not complete pooling at a finite bonus. At every finite reward, sufficiently high-cost sellers remain at each lower honest rating. In particular, the lowest label still carries information with positive probability, giving positive information value for strictly convex $B$.

If $\mu^0(z_{K-1})>\pi$, the credibility gap between the two highest labels converges to the positive number $\mu^0(z_{K-1})-\pi$, while $I_B$ converges to zero. Thus a persistent inversion is compatible with an almost uninformative rating system. The surviving lower label is accurate when encountered, but there are almost no opportunities to encounter it. The gap is a comparison conditional on seeing either label; information value also accounts for their frequencies. The gap alone is not a measure of the amount of information in the system.

The direction from which the top posterior approaches the prior requires additional care. In general, the sellers who stay below the top need not be representative of the population. If $C_a$ decreases strictly with the honest rating, however, sellers starting lower are eventually at least as likely to remain below the top. Retention is then tilted toward less favorable honest records, and the top posterior stays above the prior for all sufficiently large bonuses. The credibility gap approaches its limit from below. This is a statement about the side of convergence, not a claim that the gap must increase monotonically at every bonus.

For exponential costs, we can say more about how quickly information disappears.

\Needspace{15\baselineskip}
\begin{corollary}[The rate of information loss]\label{cor:rate}
Suppose $G(x)=1-e^{-\gamma x}$ for $x>0$ and $C_a$ is uniquely largest at the lowest honest rating $a=z_1$. Write $C_*=C_{z_1}$ and
\[
 L_*=f(z_1)\exp\left(\frac{\gamma\rho[D(\mu^0(z_1))-D(\pi)]}{C_*}\right)>0.
\]
Uniformly over equilibria, as $\Delta\to\infty$,
\begin{equation}\label{eq:rate}
 \varepsilon\sim L_*e^{-\gamma\Delta/C_*},\qquad
 I_B\sim L_*\mathcal B_B(\mu^0(z_1),\pi)e^{-\gamma\Delta/C_*},
\end{equation}
for every continuously differentiable strictly convex $B$, where
$\mathcal B_B(u,\pi)=B(u)-B(\pi)-B'(\pi)(u-\pi)>0$ for $u\ne\pi$.
\end{corollary}

The notation $x\sim y$ means that $x/y$ tends to one, so \eqref{eq:rate} is a leading approximation whose relative error vanishes as the bonus grows. Both remaining lower-rating mass and information value decay at rate $\gamma/C_*$. Once this approximation is accurate, increasing the bonus by $C_*/\gamma$ multiplies each quantity by approximately $e^{-1}$. A larger cost of reaching the top slows the loss of information; a larger exponential rate $\gamma$, which places more sellers at low manipulation costs, accelerates it.

The origin with the largest base cost retains sellers longest and dominates the remaining lower-rating mass. Under the stated ordering, that origin is the lowest honest rating. Its label is the last substantial source of information to disappear and determines the leading term.

The coefficient $\mathcal B_B(u,\pi)$ measures how much decision value at posterior $u$ exceeds a linear approximation around the prior. The approximation is $B(\pi)+B'(\pi)(u-\pi)$; subtracting it leaves the gain associated with the curvature of $B$. Geometrically, this is the vertical gap between the value function and its tangent at $\pi$. Strict convexity makes the gap positive for $u\ne\pi$.

In \eqref{eq:rate}, this gap converts the small surviving mass of lowest-rating sellers into information value. The linear part cancels because average beliefs must still equal the prior, including the small adjustment in top credibility. With uniform purchase frictions, the gap is $(u-\pi)^2/2$, giving information coefficient $L_*[\mu^0(z_1)-\pi]^2/2$. The review technology in Appendix~\ref{app:technology} supplies the required cost ordering when honest exact averages increase across displayed honest ratings and authentic review mass and review-production costs are common across honest origins.

\section{Numerical illustrations}\label{sec:example}
The purpose of the example is to distinguish three objects that an evaluation might otherwise combine: a label's credibility, its frequency, and the decision value of the whole system. The example uses seven displayed ratings, $\Z=\{4.4,4.5,\ldots,5.0\}$, with equal prior shares of high- and low-quality sellers, $\pi=1/2$. Their honest-rating distributions are
\[
 \begin{split}
 h_\HH&=(0.02,0.03,0.05,0.15,0.25,0.30,0.20),\\
 h_\LL&=(0.25,0.22,0.18,0.14,0.10,0.07,0.04).
 \end{split}
\]
Entries follow the rating grid from $4.4$ to $5.0$. For example, an honest $5.0$ occurs among $20$ percent of high-quality sellers and $4$ percent of low-quality sellers. The honest likelihood ratio is strictly increasing. Manipulation costs are exponential with rate one, $D(\mu)=\mu$, and $\rho=1$ unless otherwise stated. These are illustrative parameters, not estimates from a particular platform.

Each honest exact average is at the center of its displayed bin. Reaching a higher displayed rating $z$ requires an exact average of $z-0.05$. With authentic review mass normalized to one, fake reviews restricted to five stars, and a unit cost per fake review, the target cost is
\begin{equation}\label{eq:numericalcost}
 C(a,z)=\frac{z-0.05-a}{5.05-z}\quad(z>a),\qquad C(a,a)=0.
\end{equation}
The numerator is the required improvement in the exact average; the denominator measures how far a five-star fake review lies above the target average. As the target approaches five, each added review has less room to lift the average, so more reviews are needed. A displayed five requires eleven fake reviews per honest review from an honest $4.4$, but only one from an honest $4.9$. This simple single-score technology satisfies the main cost assumptions without requiring strict convexity of the cost of producing reviews.

\subsection{Credibility and the frequency of each rating}
Figure~\ref{fig:ratings} compares no mechanical reward, a bonus $\Delta=2$ for reaching $4.8$, and a bonus $\Delta=2$ available only at $5.0$. A bonus of two is measured in payoff units; it is not a two-star change in the rating. Panel (a) reports the high-quality probability at each label, while panel (b) reports how many sellers display it. Table~\ref{tab:equilibria} gives selected equilibrium values, including the purchase-information measure in \eqref{eq:variance}.

\begin{figure}[tbp]
\centering
\includegraphics[width=\textwidth]{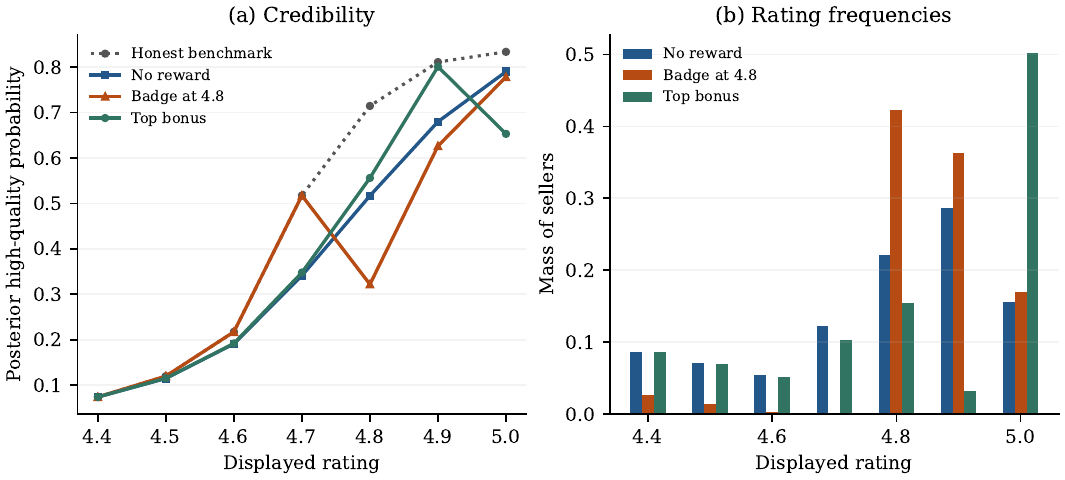}
\caption{Equilibrium credibility and rating frequencies. Panel (a) compares the honest benchmark, no mechanical reward, a reward of two for reaching 4.8, and a reward of two for displaying 5.0. Panel (b) shows the associated rating masses. The interior-threshold case illustrates inversion but does not satisfy the conservative sufficient condition \eqref{eq:interiorcondition}.}
\label{fig:ratings}
\end{figure}

\begin{table}[tbp]
\centering\small
\caption{Selected equilibria and the value of rating information}
\label{tab:equilibria}
\begin{tabular}{lrrrrr}
\toprule
Case & $\mu(4.8)$ & $\mu(4.9)$ & $\mu(5.0)$ & $m(5.0)$ & $I_{B_F}$ \\
\midrule
No reward & 0.51669 & 0.67962 & 0.78971 & 0.15654 & 0.02847 \\
Badge at 4.8 & 0.32192 & 0.62612 & 0.77809 & 0.16977 & 0.01975 \\
Top bonus & 0.55582 & 0.80065 & 0.65292 & 0.50238 & 0.02421 \\
Top bonus, $\rho=2.2$ & 0.43449 & 0.64331 & 0.64276 & 0.52743 & 0.02004 \\
Top bonus $\Delta=16$ & 0.71429 & 0.81081 & 0.52515 & 0.93155 & 0.00497 \\
Top bonus $\Delta=40$ & 0.71429 & 0.81081 & 0.50210 & 0.99476 & 0.00043 \\
Top bonus $\Delta=100$ & 0.71429 & 0.81081 & 0.50001 & 0.99998 & $1.460\times10^{-6}$ \\
\bottomrule
\end{tabular}
\end{table}

Without a mechanical reward, credibility increases with the displayed rating, as Theorem~\ref{thm:location} requires. A displayed $4.7$ has a high-quality probability of about $0.341$, below its honest value of $0.517$. The corresponding probabilities at $5.0$ are $0.790$ and $0.833$. Buyers discount ratings because they anticipate manipulation, but they still regard higher ratings as better news.

A reward of two for reaching $4.8$ changes the comparison at that threshold. Its posterior falls to about $0.322$, below the $0.517$ posterior at $4.7$. Since the reward exceeds the sufficient cleaning bound $\Lambda_j\simeq1.403$, Theorem~\ref{thm:steps}(i) ensures that all ratings below $4.8$ retain their honest meaning. The sufficient inversion condition \eqref{eq:interiorcondition} is not satisfied in this example, however. The reversal at $4.8$ is a computed outcome, demonstrating that inversion can occur beyond the cases certified by that conservative condition.

The frequency comparison is striking: only about $0.0017$ percent of sellers display $4.7$, while about $42.3$ percent display $4.8$. The lower label is relatively credible because no seller manipulates into it. It is also extremely rare because almost everyone who starts there leaves. The same reward that restores its honest meaning removes most opportunities for buyers to observe it.

A reward of two only at $5.0$ produces a posterior of about $0.653$ at the top, compared with $0.801$ at $4.9$. About $50.2$ percent of sellers display $5.0$, while $3.24$ percent display $4.9$. Corollary~\ref{cor:top} guarantees an inversion for all sufficiently large top bonuses because an honest $4.9$ indicates above-average quality. In the computed solutions, the inversion first appears at a much smaller bonus, approximately $0.65$. Numerical search does not establish uniqueness over the whole range.

For comparison, the program also evaluates rewards that rise linearly with the position of the rating on the grid. At the seven tested slopes, from $0.1$ to $5$ per step, credibility remains strictly increasing, although it is compressed relative to the honest benchmark. This is an example of an increasing reward schedule without an inversion, not a general result for all linear schedules or all slopes in that interval.

\subsection{Increasing the top bonus}
Figure~\ref{fig:information} follows the top-rating bonus over two ranges. Panels (a)--(b) show moderate bonuses, where some sellers still choose intermediate targets. Panels (c)--(d) show bonuses above the sufficient bound for the large-bonus representation.

\begin{figure}[tbp]
\centering
\includegraphics[width=0.92\textwidth]{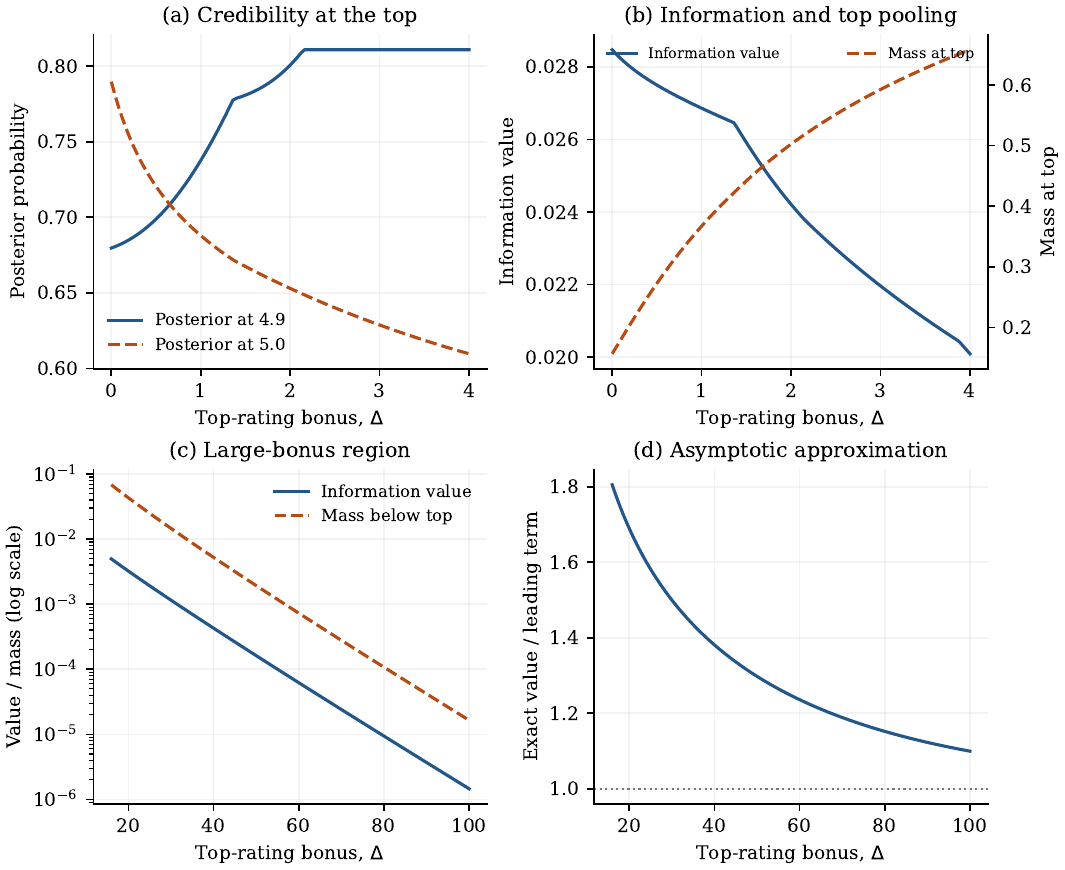}
\caption{Top rewards and information. Panels (a)--(b) display credibility at the two highest ratings, top mass, and information value for moderate bonuses. Panels (c)--(d) focus on bonuses above the sufficient bound $\overline\Delta\simeq14.836$ and display the information loss and the asymptotic approximation from Corollary~\ref{cor:rate}. The information measure is expected purchase surplus relative to observing only the prior under uniform frictions.}
\label{fig:information}
\end{figure}

In panel (a), top credibility falls while the credibility of $4.9$ rises along the computed solutions. The latter returns to its honest value once the bonus exceeds about $2.15$. Other intermediate moves persist longer. The last to disappear is a move from $4.6$ to $4.7$, at a bonus of about $14.38$. The sufficient bound $\Lambda_K\simeq14.84$ is about $3.2$ percent above this computed threshold. For the badge at $4.8$, manipulation below the threshold stops at a reward of about $1.12$, compared with the sufficient bound $\Lambda_j\simeq1.40$, about $24.8$ percent higher. The bounds are therefore informative in this example, but they are not exact critical rewards.

Panel (b) illustrates why the rating frequencies matter. Across the computed moderate-bonus grid, information value falls from about $0.0285$ at $\Delta=0$ to $0.0201$ at $\Delta=4$, while the top's population share rises from about $16$ to $66$ percent. Most of these bonuses are below the region covered by Theorem~\ref{thm:information}. The decline shown here is therefore a numerical finding rather than an application of the general monotonicity result.

For bonuses above $\overline\Delta\simeq14.84$, every equilibrium has only honest retention and top pooling. The left side of the uniqueness condition \eqref{eq:scalarunique} is below $0.093$ at this bound and decreases thereafter. Equilibrium is therefore unique throughout this region, and the theorem gives an unambiguous comparison: a larger top bonus makes ratings less informative.

The last three rows of Table~\ref{tab:equilibria} make the distinction concrete. At bonuses of $16$, $40$, and $100$, an observed $4.9$ has the same high-quality probability, about $0.811$. Meanwhile, the share displaying $5.0$ rises from about $93.2$ percent to $99.5$ percent and then $99.998$ percent. At a bonus of $100$, top credibility is approximately the prior $0.5$, and information value is only $1.460\times10^{-6}$. The credibility gap between $4.9$ and $5.0$ is still about $0.311$. The rare lower label remains favorable news, but the system almost always delivers a label that says little.

Panel (c) shows the decline of lower-rating mass and information value on a logarithmic scale. Panel (d) divides the exact information value by the leading exponential term in Corollary~\ref{cor:rate}. The ratio approaches one, as the result predicts. Convergence is relatively slow because the two largest costs of reaching the top, $C_{4.4}=11$ and $C_{4.5}=9$, are close. The ratio is still about $1.10$ at a bonus of $100$ and about $1.01$ at $200$. The lowest honest rating determines the eventual rate, but the next-lowest category continues to matter over a substantial range.

\subsection{The role of informed demand}
Finally, hold the top bonus at two and increase the weight on informed demand from $\rho=1$ to $\rho=2.2$. The inversion almost disappears: credibility at $4.9$ falls to about $0.643$, close to credibility at the top. The main change is not that honest-$4.9$ sellers stop seeking the top bonus. Instead, additional sellers from lower honest ratings find it worthwhile to target $4.9$, attracted by the more valuable return to credibility. Their entry reduces the quality share at $4.9$.

This example shows why a comparison based only on a seller already near the top can be misleading. Increasing the return to credibility affects every seller who can reach a credible target, including those starting much lower. In the computed example, their entry nearly eliminates the inversion by reducing the credibility of the lower label. This is a numerical finding, not a general comparative-static theorem.

The accompanying program calculates exact choice intervals for the cost distribution, solves for Bayesian consistency, and checks best responses and probability accounting. In the large-bonus region it compares the full-grid and scalar solutions, verifies the garbling matrices, and checks the asymptotic coefficient. It also uses bisection to locate the disappearance of inflow along the computed solutions and evaluates the stated linear reward schedules. Multiple initial guesses are used to check numerical robustness; they are not a proof of uniqueness.

\section{Implications for rating policy}\label{sec:policy}
\subsection{Evaluate composition and frequency together}
The model gives a reason to assess a rating policy using more than the average displayed score or the credibility of its best-looking categories. A rise in average displayed ratings can reflect movement toward a rewarded label, even though underlying quality is unchanged. A lower category can become more credible when inflow stops, yet contribute less to information because its remaining population is small. Category credibility and category frequency answer different evaluation questions.

For a platform applying a threshold reward, a useful empirical evaluation would therefore examine the joint distribution of displayed ratings and independently measured quality. A histogram can show concentration at the threshold, but cannot establish whether that concentration reflects manipulation or genuine quality improvement. Conversely, a quality comparison between $4.9$ and $5.0$ cannot reveal how much information the system provides without the frequency of each label. In the purchase interpretation, these two ingredients enter the surplus measure in \eqref{eq:information} directly.

The strongest theoretical comparison applies to larger top bonuses in the large-bonus region. Along a fixed extremal equilibrium selection, the stronger reward erases information useful for every decision problem. This is an information-cost result under fixed sampling and fixed quality. It does not show that a platform should eliminate all certification rewards. Rewards could also induce better service, attract entry, or change buyer attention, none of which is modeled here.

The large-bonus limit is likewise an experiment in reward strength, rather than a policy that can be financed by any fixed budget. If a top reward costs one resource unit per payoff unit, its aggregate incremental cost is $\Delta m(z_K)$. Since $m(z_K)\to1$, that cost grows without bound. The limit identifies what an increasingly strong incentive does to information; it does not establish that a budget-constrained platform can or should implement it. The following allocation keeps the available exposure fixed.

\subsection{Allocating exposure according to credibility}\label{subsec:allocation}
The preceding results take the reward schedule as given. We now consider a platform that controls a total amount $N>0$ of additional exposure, with each unit adding one unit to a seller's payoff. The platform can make exposure depend on the quality information in a rating rather than attach it directly to the numerical label.

Let $A:[0,1]\to(0,\infty)$ be continuous and nondecreasing. Consider the allocation per seller
\begin{equation}\label{eq:allocation}
 \alpha(z)=\frac{N A(\mu(z))}{\sum_{z'}m(z')A(\mu(z'))}.
\end{equation}
The numerator assigns each seller a weight $A(\mu(z))$ based on the quality information in its rating. The denominator adds up these weights over the population. Dividing by that total and multiplying by $N$ allocates exactly the available exposure: $\sum_zm(z)\alpha(z)=N$. A constant $A$ gives equal exposure; an increasing $A$ directs more exposure per seller toward more credible categories. The normalization is essential because both category sizes and beliefs change when sellers manipulate.

The platform announces the formula, and sellers anticipate the exposure implied by equilibrium beliefs and rating frequencies. This allocation replaces the exogenous mechanical reward $r$; no separate reward for the numerical label is added.

\begin{proposition}[Credibility-linked exposure]\label{prop:allocation}
An equilibrium under \eqref{eq:allocation} exists, the exposure budget is satisfied, and every equilibrium has strictly increasing credibility on the rating grid.
\end{proposition}

For example, $A(\mu)=1+\mu$ gives a seller a weight between one and two, with higher-quality beliefs receiving larger weights. This does not prescribe the best choice of $A$; it shows how the rule can be implemented as a normalized allocation. At equilibrium, the denominator of \eqref{eq:allocation} is common to all sellers. Both exposure and informed demand then increase with credibility. A less credible higher rating cannot compensate an entrant for its greater cost, so the location theorem's dominance argument rules out an inversion. The rule is nevertheless an equilibrium object: exposure depends on beliefs, beliefs depend on manipulation, and manipulation responds to exposure. The proof solves jointly for beliefs and the normalizing denominator.

This proposition is a design illustration under the assumption that equilibrium credibility is available to the platform. It establishes that such a rule is consistent with equilibrium and a fixed exposure budget. Implementation would also require a way to estimate quality probabilities; the proposition supplies neither that statistical procedure nor a proof that the rule is optimal.

The allocation concerns the \emph{amount} of exposure, not simply the order of a list. When credibility increases with ratings, sorting by credibility and sorting by the displayed number give the same ranking. The two policies can still create different incentives because the exposure gain from moving one step need not be the same. With a fixed set of sellers and beliefs, shifting a unit of exposure toward a seller with greater expected buyer value improves a linear exposure-weighted objective. A full change of policy is more complicated because seller behavior and beliefs adjust.

The information result identifies a separate cost of large top bonuses: in the large-bonus region, they reduce the value of ratings for every decision problem along the specified equilibrium selections. The exposure proposition does not show that every credibility-linked rule is more informative or better for welfare. An increasing posterior schedule can still be compressed by manipulation. This distinction is related to the commitment problem emphasized by \citet{frankel2022}. A welfare comparison would additionally require manipulation costs, reward costs, and possible quality investment induced by certification \citep{hui2023}.

\subsection{Implementation and limits of the mechanism}
The exposure rule requires credible estimates of quality probabilities. A platform could investigate independent audits, verified transactions, complaint outcomes, or other quality information, but the model does not establish that any particular estimator is sufficient. Using the manipulated star score itself as an independent measure would simply repeat the original problem. In the theoretical rule, buyers and the platform agree on equilibrium credibility, and sellers anticipate the announced allocation formula.

An empirical test of the reward mechanism would require a change in the reward rule, rating frequencies before and after that change, and a quality measure not mechanically constructed from the displayed score. The model predicts movement toward a rewarded threshold and, in the large-top-bonus region, a decline in the frequency of every lower category along the specified selections. The motivating purchase pattern is insufficient on its own: prices, exposure, review counts, and quality investment can also change. These are directions for identification, not an identification theorem supplied by the paper.

Conditional quality neutrality remains the central restriction. It allows the response to a reward to be represented as a quality-independent transition from honest to displayed categories. Different manipulation costs by quality within an origin could add information through selection. Review counts or histories could also allow buyers to distinguish sellers pooled at the same score. Finally, the return function holds sale terms and margins fixed. The results therefore establish an informational cost of score-based incentives under stated conditions; they do not determine equilibrium prices, total consumption, or total welfare.

\section{Conclusion}\label{sec:conclusion}
A reward for a high displayed rating can weaken the information that ratings are intended to provide. In the large-top-bonus region, raising the bonus makes the equilibrium rating system less informative in Blackwell's order along either extremal equilibrium selection. Lower labels preserve their honest quality composition while appearing less often. In the limit, almost everyone displays the top rating and its credibility approaches the population average. A nearly perfect rating can still be more credible than a perfect one, even though the system supplies almost no information.

The supporting results identify why the separate reward matters. Under upward, conditionally quality-neutral manipulation, credibility must increase between adjacent ratings unless the displayed number earns an additional reward. A sufficiently large reward step removes inflow into lower categories and pools sellers at the threshold. Rational buyers account for the changed composition, but cannot reconstruct the honest origins erased by pooling.

For platform policy, this distinction makes category frequency as relevant as category credibility. A budget-feasible exposure rule linked to posterior quality can avoid reversals, although monotone credibility alone does not establish optimal information or welfare. Extending the analysis to quality investment, endogenous prices, additional public information, and detection would allow these informational effects to be assessed alongside other consequences of platform rewards.

\clearpage
\appendix
\numberwithin{equation}{section}
\numberwithin{lemma}{section}
\numberwithin{proposition}{section}
\section{Equilibrium foundations}\label{app:foundations}
\subsection{Choice intervals and equilibrium existence}
We first describe who chooses each rating. Define
\[
 R=\max_zr(z)-\min_zr(z),\qquad \Cmin=\min_{a<z}C(a,z)>0.
\]
The largest possible gain in gross returns is at most $R+\rho\DD$, while every upward move has base cost at least $\Cmin$.

\Needspace{9\baselineskip}
\begin{lemma}[Choice thresholds]\label{lem:thresholds}
For a fixed belief vector $\mu\in[0,1]^K$, its induced return vector $V(z)=r(z)+\rho D(\mu(z))$, and honest rating $a$, optimal displayed ratings are nonincreasing in $\kappa$. Rating $z\geq a$ is optimal exactly on the possibly empty interval $[\underline\kappa(a,z),\overline\kappa(a,z)]\cap(0,\infty)$, where
\begin{align}
 \underline\kappa(a,z)&=\max\left\{0,\max_{z'>z}\frac{V(z')-V(z)}{C(a,z')-C(a,z)}\right\},\label{eq:lowercutoff}\\
 \overline\kappa(a,z)&=\min_{a\leq z'<z}\frac{V(z)-V(z')}{C(a,z)-C(a,z')}.\label{eq:uppercutoff}
\end{align}
The conventions are $\underline\kappa(a,z_K)=0$ and $\overline\kappa(a,a)=+\infty$. Finite cutoffs are continuous in $V$, ties occur at finitely many cost parameters, and
\begin{equation}\label{eq:residents}
 p(a,a)\geq1-G\left(\frac{R+\rho\DD}{\Cmin}\right)>0.
\end{equation}
\end{lemma}

The two formulas describe the cost range in which a target is worth choosing. Each ratio is an extra gross return divided by an extra base cost. If moving from one target to another adds ten units of return and five units of base cost, the move is worthwhile when $\kappa\leq2$. The same comparison is made against every feasible alternative.

To choose $z$, a seller must have costs high enough to reject every more ambitious target; the most demanding of these comparisons gives $\underline\kappa(a,z)$. Its costs must also be low enough to prefer $z$ to every cheaper target; the tightest comparison gives $\overline\kappa(a,z)$. If the lower cutoff exceeds the upper cutoff, no cost parameter makes $z$ the best choice. As manipulation becomes more expensive, the chosen rating can only fall or stay the same.

The bound in \eqref{eq:residents} ensures that every category keeps some honest residents. A seller with a sufficiently high $\kappa$ finds even the largest possible reward gain too small to justify the cheapest upward move. Because the cost distribution has an unbounded upper tail, such sellers always have positive mass.

\begin{proposition}[Equilibrium existence]\label{prop:existence}
Under Assumptions~\ref{ass:honest}--\ref{ass:cost}, an equilibrium exists. Every rating is displayed by a positive mass of each quality. For each honest rating, the optimal displayed rating is unique for $G$-almost every $\kappa$.
\end{proposition}

Positive honest masses mean that every displayed rating is observed, so Bayes' rule determines beliefs everywhere. The existence proof starts with candidate beliefs, computes sellers' responses, and then computes the beliefs those responses justify. This map is continuous because there are no mass points at choice cutoffs. Brouwer's fixed-point theorem gives a belief vector that reproduces itself. General equilibrium uniqueness is not required for the composition and location results that follow.

\subsection*{Proof of Lemma~\ref{lem:thresholds}}
The proof first orders choices by manipulation cost, then derives the interval of costs supporting each target. Finally, it shows that sufficiently high-cost sellers stay honest.

For $z'>z\geq a$ and $\kappa'>\kappa$, the difference between the payoffs of $z'$ and $z$ changes by
\[
 -(\kappa'-\kappa)[C(a,z')-C(a,z)]<0.
\]
If a larger rating were chosen at a larger cost parameter, adding the two optimality inequalities would contradict this strict inequality. Hence every optimal selection is nonincreasing in $\kappa$.

For each $z'>z$, the inequality $V(z)-\kappa C(a,z)\geq V(z')-\kappa C(a,z')$ is equivalent to the lower bound on $\kappa$ in \eqref{eq:lowercutoff}. Each feasible $z'<z$ gives the upper bound in \eqref{eq:uppercutoff}. Their intersection is the stated interval, which is empty if its upper endpoint is below its lower endpoint. The finite endpoints are maxima and minima of finitely many affine functions of $V$, so they are continuous. Distinct ratings can tie at only one $\kappa$ per pair because their cost slopes differ. There are finitely many pairs, and continuity of $G$ makes this a null set.

Finally, for $z>a$,
\[
 \Pi(z;a,\kappa)-\Pi(a;a,\kappa)
 \leq R+\rho\DD-\kappa\Cmin.
\]
Every $\kappa>(R+\rho\DD)/\Cmin$ chooses the honest rating. Its probability is the positive lower bound in \eqref{eq:residents}.

\subsection*{Proof of Proposition~\ref{prop:existence}}
We construct a map from proposed beliefs to the beliefs justified by optimal seller choices. Positive honest masses keep the Bayes denominators away from zero, allowing a continuous fixed-point argument.

For a candidate belief vector $\mu\in[0,1]^K$, form $V_\mu(z)=r(z)+\rho D(\mu(z))$ and the cutoffs in Lemma~\ref{lem:thresholds}. Optimal transition probabilities are
\[
 p_\mu(a,z)=\max\{0,G(\overline\kappa_\mu(a,z))-G(\underline\kappa_\mu(a,z))\}.
\]
The extension of $G$ to nonpositive arguments handles negative upper cutoffs. For each positive cost parameter a maximizer exists because the feasible target set is finite. The choice intervals consequently cover the positive cost line and overlap only at finitely many tie points, so $p_\mu(a,\cdot)$ is a probability distribution over feasible ratings.

The functions $p_\mu$ are continuous in $\mu$. By \eqref{eq:residents}, $\phi_{\theta,\mu}(z)\geq h_\theta(z)[1-G((R+\rho\DD)/\Cmin)]>0$, uniformly over candidates. Bayes' rule therefore defines a continuous map from $[0,1]^K$ into itself. Brouwer's theorem yields a fixed point. Choosing the smallest optimal rating at each $(a,\kappa)$ gives a measurable strategy: each choice region is an intersection of finitely many closed half-lines, with finitely many boundaries removed by the tie convention. Its transitions and the fixed-point beliefs satisfy Definition~\ref{def:equilibrium}. Positive masses and almost-everywhere uniqueness follow from the threshold lemma.

\section{Proofs of the equilibrium results}\label{app:proofs}
\subsection*{Proof of Lemma~\ref{lem:composition}}
Conditional on an honest origin $a$, both qualities have the same cost distribution and maximization problem. Their optimal actions agree except on the finite tie set, which has probability zero. Thus the transition probabilities $p(a,z)$ are common across qualities. Since $f(a)\mu^0(a)=\pi h_\HH(a)$,
\[
 \mu(z)=\frac{\pi\sum_{a\leq z}h_\HH(a)p(a,z)}{m(z)}
       =\sum_{a\leq z}\frac{f(a)p(a,z)}{m(z)}\mu^0(a).
\]
The weights sum to one by \eqref{eq:distributions}, and $p(z,z)>0$ by Lemma~\ref{lem:thresholds}. Strict increase of the honest posteriors gives $\mu(z)\leq\mu^0(z)$, with equality exactly when all lower-origin weights are zero. There is no possible lower origin at $z_1$.

For later use, the equivalent likelihood-ratio representation is
\[
 \ell(z)=\sum_{a\leq z}\omega_z(a)\ell^0(a),\qquad
 \omega_z(a)=\frac{h_\LL(a)p(a,z)}{\sum_{a'\leq z}h_\LL(a')p(a',z)}.
\]
The strictly increasing function $t\mapsto\pi t/(\pi t+1-\pi)$ converts likelihood ratios to posteriors. Both representations place every posterior between $\mu^0(z_1)$ and $\mu^0(z_K)$. Finally, the common stochastic matrix $p$ transforms $h_\theta$ into $\phi_\theta$ for both qualities. This is a Blackwell garbling: after seeing the honest rating, one can generate the displayed rating without observing quality.

\subsection*{Proof of Theorem~\ref{thm:location}}
We argue by contradiction. A proposed reversal with no increase in the reward would remove all entrants to the higher category; its honest composition would then rule out the reversal.

Suppose $r(z_k)\leq r(z_{k-1})$ and $\mu(z_k)\leq\mu(z_{k-1})$. Then $V(z_k)\leq V(z_{k-1})$. For every $a<z_k$ and $\kappa>0$, both targets are feasible and
\[
 \Pi(z_{k-1};a,\kappa)-\Pi(z_k;a,\kappa)
 =V(z_{k-1})-V(z_k)+\kappa[C(a,z_k)-C(a,z_{k-1})]>0.
\]
Thus $z_k$ receives no inflow. Lemma~\ref{lem:composition} gives
\[
 \ell(z_k)=\ell^0(z_k)>\ell^0(z_{k-1})\geq\ell(z_{k-1}),
\]
contradicting the proposed posterior ordering. Apply the implication to each adjacent pair of a constant-reward block for the remaining statements.

\subsection*{Proof of Corollary~\ref{cor:ends}}
The bottom comparison follows from the origins that can reach it. At the top, we consider separately whether its credibility exceeds or falls below that of the preceding category.

The lowest posterior statement is Lemma~\ref{lem:composition}. At $z_2$, composition mixes only the honest likelihood ratios at $z_1$ and $z_2$, with strictly positive weight on the latter. Hence $\ell(z_2)>\ell^0(z_1)=\ell(z_1)$.

For the top, suppose $r$ is nondecreasing. If $\mu(z_K)\geq\mu(z_{K-1})$, then $V(z_K)\geq V(z_{K-1})$, with equality requiring equal rewards and equal posteriors. Theorem~\ref{thm:location} rules out that equality, so the return difference is strictly positive. A positive mass of sufficiently low-cost sellers at $z_{K-1}$ chooses $z_K$. If instead $\mu(z_K)<\mu(z_{K-1})$, no inflow at $z_K$ would imply $\mu(z_K)=\mu^0(z_K)>\mu^0(z_{K-1})\geq\mu(z_{K-1})$, a contradiction. In either case the top receives inflow and its posterior is strictly diluted.

\subsection*{Proof of Theorem~\ref{thm:steps}}
The three parts establish different claims: intermediate targets lose their entrants, the rewarded target receives sellers from every lower origin, and a bound on its resulting composition guarantees inversion.

\paragraph{Part (i): a uniform cleaning bound.}
Put $d_a=D(\mu^0(a))$ and $d_1=D(\mu^0(z_1))$. Define
\begin{equation}\label{eq:cleanbound}
 \Lambda_j=\rho\max_{a<z<z_j}
 \left\{d_z-d_1+\left[\frac{C(a,z_j)}{C(a,z)}-1\right](d_z-d_a)\right\},
\end{equation}
where a maximum over an empty set is zero. The maximum is finite because the grid is finite and every denominator is strictly positive.

We prove the absence of inflow below $z_j$ by induction over the target rating. The bottom rating receives no inflow and has posterior $\mu^0(z_1)$. Suppose all ratings below a given $z<z_j$ receive no inflow, so their posteriors equal their honest values. Consider an origin $a<z$. If a seller chooses $z$ rather than its own honest rating, then
\[
 \kappa\leq\frac{\rho[D(\mu(z))-D(\mu^0(a))]}{C(a,z)}
 \leq\frac{\rho(d_z-d_a)}{C(a,z)}.
\]
The first inequality uses the induction hypothesis $\mu(a)=\mu^0(a)$, which is where the hypothesis is needed, and the second uses $\mu(z)\leq\mu^0(z)$ from Lemma~\ref{lem:composition}. Choosing $z$ rather than the rewarded threshold $z_j$ also requires
\[
 \kappa\geq\frac{\Delta+\rho[D(\mu(z_j))-D(\mu(z))]}{C(a,z_j)-C(a,z)}
 \geq\frac{\Delta-\rho(d_z-d_1)}{C(a,z_j)-C(a,z)}.
\]
All posteriors are at least $\mu^0(z_1)$. For $\Delta>\Lambda_j$, the last lower bound strictly exceeds the previous upper bound, by \eqref{eq:cleanbound}. No such seller can choose $z$. Thus $z$ receives no inflow and has honest posterior. Induction proves the claim.

\paragraph{Part (ii).}
Take $\Delta>\rho\DD$ and define $\kappa_\Delta=(\Delta-\rho\DD)/\Cmax$. For any $a\leq z_j$, the lower choice cutoff for $z_j$ is at most $\beta_j$: ratings above $z_j$ have the same mechanical reward and can add at most $\rho\DD$ to gross return, while the first further cost increment is at least $c_j$. For $a<z_j$, the upper cutoff against every target below $z_j$ is at least $\kappa_\Delta$, since the return advantage is at least $\Delta-\rho\DD$ and the cost increment is at most $\Cmax$. For $a=z_j$ the upper cutoff is infinite. Hence
\begin{equation}\label{eq:inflowbound}
 p(a,z_j)\geq [G(\kappa_\Delta)-G(\beta_j)]_+,\qquad a\leq z_j.
\end{equation}
As $\Delta\to\infty$, $G(\kappa_\Delta)\to1>G(\beta_j)$. A finite $\Delta_1\geq\max\{\Lambda_j,\rho\DD\}$ makes the right side positive for all larger rewards. Every lower origin then contributes positive inflow. Strict dilution follows from Lemma~\ref{lem:composition}.

\paragraph{Part (iii).}
For $\Delta>\Delta_1$, the rating just below the threshold is honest. Using \eqref{eq:inflowbound} in the denominator and $p(a,z_j)\leq1$ in the numerator gives
\[
 \ell(z_j)\leq
 \frac{H_\HH(z_j)}{H_\LL(z_j)[G(\kappa_\Delta)-G(\beta_j)]}.
\]
Under \eqref{eq:interiorcondition}, this bound is smaller than $\ell^0(z_{j-1})$ for all sufficiently large $\Delta$. The bound is independent of equilibrium, so a common finite $\Delta_2$ works for every equilibrium. Bayes' rule converts the likelihood-ratio inequality to the stated posterior inversion.

\subsection*{Proof of Corollary~\ref{cor:top}}
At the top there is no onward movement, so the sufficient condition for an interior threshold simplifies directly.

At $j=K$, $H_\HH(z_K)=H_\LL(z_K)=1$ and $G(\beta_K)=G(0)=0$. Condition~\eqref{eq:interiorcondition} becomes $1<\ell^0(z_{K-1})$. Apply Theorem~\ref{thm:steps}(iii) for inversion and part (i) for honest lower posteriors.

\subsection*{Proof of Lemma~\ref{lem:scalar}}
We first derive the one-dimensional equation from any equilibrium, then show that each root supports optimal choices in the full rating grid. The final step bounds demand feedback to establish uniqueness.

For $\Delta>\overline\Delta$, no seller enters an intermediate lower rating in equilibrium. Thus a seller with $a<z_K$ chooses between $a$ and $z_K$, while an honest top seller stays at the top. The difference in gross returns is $v-\rho D(\mu^0(a))$, yielding the retention probability \eqref{eq:survival}. The retained population masses give \eqref{eq:epsilon}. Subtracting their high-quality masses from the total high-quality population gives \eqref{eq:topbelief}. Consistency of $v$ gives \eqref{eq:scalar}.

Conversely, consider any scalar root. Assign the retention probabilities in \eqref{eq:survival}, the honest posteriors to lower ratings, and $b(v)$ to the top. These beliefs are Bayesian for the proposed transitions, and $b(v)$, an average of honest posteriors with weights proportional to $f(z_K)$ on $\mu^0(z_K)$ and $f(a)[1-s_a(v)]$ on each $\mu^0(a)$ with $a<z_K$, lies between the lowest and highest honest posteriors. The choice between $a$ and $z_K$ is optimal within that pair by construction. No target $a<z<z_K$ can be optimal in the full choice set: the two bounds in the proof of Theorem~\ref{thm:steps}(i), now with honest posteriors at every lower rating, contradict each other for $\Delta>\Lambda_K$. Thus the proposed stay-or-top strategy is globally optimal, including tie points after a consistent choice convention, and the scalar root gives an equilibrium.

For every real $v$, $0<s_a(v)\leq1$ and $1-\varepsilon(v)\geq f(z_K)>0$. The numerator in $b(v)$ equals the high-quality mass of honest top sellers plus the high-quality mass departing lower origins; the complementary numerator has the analogous low-quality interpretation. Hence $b(v)\in(0,1)$ for all $v$, and $b$ is continuous. The function $\Psi(v)=v-\rho D(b(v))$ tends to negative infinity as $v\to-\infty$ and to positive infinity as $v\to\infty$, since $D$ is bounded. The intermediate value theorem gives a root. The root set is closed by continuity and bounded by the interval stated next, so it is compact. Every root lies in the compact interval $[\Delta+\rho D(0),\Delta+\rho D(1)]$.

Now let $G$ have a density $g$ continuous on $(0,\infty)$ and let $D$ be Lipschitz with constant $L_D$. On the root interval every cutoff argument $(v-\rho D(\mu^0(a)))/C_a$ is at least $(\Delta-\rho\DD)/\Cmax>0$. Differentiating \eqref{eq:survival}--\eqref{eq:topbelief} gives
\[
 b'(v)=\frac{1}{1-\varepsilon(v)}\sum_{a<z_K}\frac{f(a)}{C_a}\,
 g\!\left(\frac{v-\rho D(\mu^0(a))}{C_a}\right)[\mu^0(a)-b(v)] .
\]
Since $|\mu^0(a)-b(v)|\leq1$ and $1-\varepsilon(v)\geq f(z_K)$,
\begin{equation}\label{eq:bprimebound}
 |b'(v)|\leq\frac{1}{f(z_K)}\sum_{a<z_K}\frac{f(a)}{C_a}\,
 g\!\left(\frac{v-\rho D(\mu^0(a))}{C_a}\right)
 \qquad\text{on }[\Delta+\rho D(0),\Delta+\rho D(1)].
\end{equation}
The map $\Phi(v)=\Delta+\rho D(b(v))$ takes the root interval into itself and is Lipschitz there with constant $\rho L_D\sup|b'|$. When this constant is below one, $\Phi$ is a contraction and its unique fixed point is the only scalar root, hence the only equilibrium. For part (i), exponential costs give $g(x)=\gamma e^{-\gamma x}$, and the bound \eqref{eq:bprimebound} is at most $(\gamma/f(z_K))\sum_af(a)C_a^{-1}\exp[-\gamma(\Delta-\rho\DD)/C_a]$, so \eqref{eq:scalarunique} implies the contraction condition. For part (ii), the arguments of $g$ in \eqref{eq:bprimebound} exceed $(\Delta-\rho\DD)/\Cmax$, which tends to infinity with $\Delta$; since $g(x)\to0$, the right-hand side of \eqref{eq:bprimebound} tends to zero uniformly on the root interval, and the contraction condition holds for all $\Delta$ beyond some $\Delta_u\geq\overline\Delta$.

\subsection*{Proof of Theorem~\ref{thm:information}}
The proof orders the smallest and largest roots, constructs the garbling between the associated experiments, and then establishes the limiting loss of information. The last argument determines the side from which top credibility approaches the prior under the additional cost ordering.

\paragraph{Part (i): root ordering.}
Write $\Psi(v)=v-\rho D(b(v))$ as above, and take $\Delta_2>\Delta_1>\overline\Delta$. Below the smallest root $v_-(\Delta_1)$, $\Psi(v)<\Delta_1$; otherwise its negative left tail and continuity would imply an earlier root. At that smallest root $\Psi=\Delta_1<\Delta_2$ as well. No root for $\Delta_2$ can be weakly below it, so $v_-(\Delta_2)>v_-(\Delta_1)$. At the largest old root, $\Psi=\Delta_1<\Delta_2$, and its positive right tail ensures a new root strictly to the right. Hence $v_+(\Delta_2)>v_+(\Delta_1)$. This proof allows a nonmonotone $\Psi$ and multiple roots.

Because $S$ is nonincreasing, $s_a(v_2)\leq s_a(v_1)$ on either selection. All these probabilities are strictly positive at finite arguments. Define a stochastic matrix $Q$ on displayed labels by
\[
 Q(a,a)=\frac{s_a(v_2)}{s_a(v_1)},\qquad
 Q(a,z_K)=1-Q(a,a)\quad(a<z_K),\qquad Q(z_K,z_K)=1,
\]
with all other entries zero. At a lower label, the first experiment has conditional mass $h_\theta(a)s_a(v_1)$, so its image under $Q$ has mass $h_\theta(a)s_a(v_2)$. The remaining mass is transferred to the top, exactly as in the second experiment. Therefore $\phi_{\theta,2}=\phi_{\theta,1}Q$ for both qualities. The second experiment is a Blackwell garbling of the first. For completeness, draw an old label $Z_1$ from the first experiment and then a new label $Z_2$ using $Q$. The Markov relation $\theta\to Z_1\to Z_2$ implies
\[
 \mu_2(Z_2)=\E[\mu_1(Z_1)\mid Z_2].
\]
Conditional Jensen's inequality therefore gives $\E[B(\mu_2(Z_2))]\leq\E[B(\mu_1(Z_1))]$. Subtracting the same $B(\pi)$ proves $I_{B,2}\leq I_{B,1}$.

\paragraph{Part (ii): strictness.}
At a root and for every $a<z_K$,
\[
 v-\rho D(\mu^0(a))
 =\Delta+\rho[D(b(v))-D(\mu^0(a))]
 \geq\Delta-\rho\DD>0.
\]
If $G$ is strictly increasing on positive arguments, strictly increasing $v$ makes every retention probability strictly smaller. In particular, the garbling moves positive mass from $z_1$ into the top. The old posterior at $z_1$ is $\mu^0(z_1)$, whereas the old top posterior is strictly larger: the top includes positive mass of honest top residents, and all origins have honest posteriors at least $\mu^0(z_1)$. The new top therefore pools positive masses with different old posteriors. The event $Z_2=z_K$ has positive probability, and conditional on this event the old posterior is nonconstant: it can come from the old lowest label or the old top label, each with positive probability. Strict convexity makes conditional Jensen's inequality strict on this event, and hence makes the unconditional information comparison strict.

\paragraph{Part (iii): collapse.}
For every equilibrium at every sufficiently large bonus,
\begin{equation}\label{eq:tailbound}
 0<\varepsilon\leq[1-f(z_K)]S\left(\frac{\Delta-\rho\DD}{\Cmax}\right)\longrightarrow0.
\end{equation}
The bound follows from the previous cutoff lower bound and $C_a\leq\Cmax$. It is uniform over equilibria. Bayesian plausibility gives
\[
 \mu(z_K)-\pi
 =\frac{\sum_{a<z_K}f(a)s_a(v)[\pi-\mu^0(a)]}{1-\varepsilon},
 \qquad |\mu(z_K)-\pi|\leq\frac{\varepsilon}{1-\varepsilon}.
\]
Thus the top posterior converges uniformly to the prior. Continuity and boundedness of $B$ imply that \eqref{eq:information} converges to zero: the lower-label contribution vanishes with $\varepsilon$, and the top contribution tends to $B(\pi)$. For $\Delta>\overline\Delta$ the rating $z_{K-1}$ receives no inflow and $\mu(z_{K-1})=\mu^0(z_{K-1})$, so the credibility gap $\mu(z_{K-1})-\mu(z_K)$ converges to $\mu^0(z_{K-1})-\pi$, uniformly over equilibria.

\paragraph{Part (iii), continued: the side of convergence.}
The remaining claim requires a comparison of who stays below the top. Let $C_a$ be strictly decreasing in $a<z_K$ and let $v$ be any root. For $a<a'<z_K$,
\[
 \frac{v-\rho D(\mu^0(a'))}{C_{a'}}-\frac{v-\rho D(\mu^0(a))}{C_a}
 =v\left(\frac1{C_{a'}}-\frac1{C_a}\right)-\rho\left(\frac{D(\mu^0(a'))}{C_{a'}}-\frac{D(\mu^0(a))}{C_a}\right),
\]
and the first term grows without bound with $v\geq\Delta+\rho D(0)$ while the second is bounded. Hence, for all sufficiently large $\Delta$, the cutoff argument is increasing in $a$ at every root, and $s_a(v)$ is nonincreasing in $a$. Write $w(a)=\pi-\mu^0(a)$, which is strictly decreasing in $a$, and note that $\sum_{a<z_K}f(a)w(a)=f(z_K)[\mu^0(z_K)-\pi]>0$ by Bayesian plausibility of the honest posteriors. Let $q_a=f(a)/[1-f(z_K)]$ for $a<z_K$. These positive weights sum to one. The sequences $s_a$ and $w(a)$ are both nonincreasing, so
\[
 \Cov_q(s,w)
 =\frac12\sum_{a,a'<z_K}q_aq_{a'}(s_a-s_{a'})(w(a)-w(a'))\geq0.
\]
Consequently,
\[
 \sum_{a<z_K}f(a)s_aw(a)
 \geq\frac{\sum_{a<z_K}f(a)s_a\;\sum_{a<z_K}f(a)w(a)}{1-f(z_K)}>0.
\]
The strict last inequality follows from positive retention and the strictly positive honest-composition sum above; it does not require strict ordering of the retention probabilities. The left side equals $(1-\varepsilon)[\mu(z_K)-\pi]$, so $\mu(z_K)>\pi$.

\paragraph{Part (iii), continued: finite bonuses.}
At every finite bonus, Lemma~\ref{lem:thresholds} leaves positive mass at the lowest label. Its posterior is $\mu^0(z_1)<\pi$, because the strictly ordered honest posteriors average to $\pi$. Bayesian plausibility implies that the posterior distribution cannot be constant. Strict Jensen's inequality gives $I_B>0$ for strictly convex $B$. This argument applies even outside the large-bonus region. Honest meaning at every lower label within that region follows from Theorem~\ref{thm:steps}(i).

\subsection*{Proof of Corollary~\ref{cor:rate}}
With exponential costs, retention separates into a term governing the decay rate and a term converging to a positive constant. The most expensive honest origin to move to the top determines the leading mass; an expansion of decision value gives the corresponding information term.

For sufficiently large $\Delta$ all cutoffs are positive. Under exponential costs,
\[
 s_a=\exp\left(-\frac{\gamma\Delta}{C_a}\right)
 \exp\left(\frac{\gamma\rho[D(\mu^0(a))-D(\mu(z_K))]}{C_a}\right).
\]
By continuity of $D$ and Theorem~\ref{thm:information}(iii), the second factor converges to the same expression with $\mu(z_K)$ replaced by $\pi$, uniformly over equilibria. The unique largest $C_a=C_*$ belongs to $z_1$, so every other retention probability is of smaller order than $e^{-\gamma\Delta/C_*}$. Summing $f(a)s_a$ gives the first equivalence in \eqref{eq:rate}.

Let $u_a=f(a)s_a$, $b=\mu(z_K)$, and $\varepsilon=\sum u_a$. Since $b-\pi=O(\varepsilon)$ and $B$ is differentiable at $\pi$,
\[
 \begin{split}
 I_B&=\sum_{a<z_K}u_a[B(\mu^0(a))-B(\pi)]
       +(1-\varepsilon)[B(b)-B(\pi)]\\
 &=\sum_{a<z_K}u_a
 \{B(\mu^0(a))-B(\pi)-B'(\pi)[\mu^0(a)-\pi]\}+o(\varepsilon),
 \end{split}
\]
where the last equality uses $(1-\varepsilon)(b-\pi)=\sum_a u_a[\pi-\mu^0(a)]$. To make the remainder uniform, fix $\eta>0$. Differentiability at $\pi$ ensures $|B(b)-B(\pi)-B'(\pi)(b-\pi)|\leq\eta|b-\pi|$ whenever $b$ is sufficiently close to $\pi$. Uniform convergence of $b$ puts every equilibrium in this neighborhood beyond a common bonus. Since $(1-\varepsilon)|b-\pi|\leq\varepsilon$, the absolute remainder above is then at most $\eta\varepsilon$, uniformly over equilibria. Only the lowest-origin term survives at the leading order. Strict convexity gives $\mathcal B_B(\mu^0(z_1),\pi)>0$, proving the second equivalence.

\subsection*{Proof of Proposition~\ref{prop:allocation}}
Because exposure depends on both beliefs and rating frequencies, the fixed point must include the normalizing denominator as well as beliefs. After establishing existence, the credibility comparison follows from the same dominance argument as before.

Let $A_{\min}=\min_{[0,1]}A>0$ and $A_{\max}=\max_{[0,1]}A$. On the compact convex set $[0,1]^K\times[A_{\min},A_{\max}]$, consider candidate beliefs $\mu$ and normalizer $L$. The seller's gross return is
\[
 V_{\mu,L}(z)=\frac{N A(\mu(z))}{L}+\rho D(\mu(z)).
\]
It is continuous and uniformly bounded over this domain. For each fixed candidate pair $(\mu,L)$, apply the cutoff formulas in Lemma~\ref{lem:thresholds} with $r(z)=NA(\mu(z))/L$ held fixed at that pair. Since $V_{\mu,L}$ varies continuously with $(\mu,L)$, the cutoffs and, by continuity of $G$, the resulting transition probabilities are continuous in $(\mu,L)$. In particular, the range of $V_{\mu,L}$ is at most
\[
 M_V=\frac{N(A_{\max}-A_{\min})}{A_{\min}}+\rho\DD<\infty.
\]
Every seller with $\kappa>M_V/\Cmin$ stays at its honest rating. Thus the honest-resident probability is uniformly at least $1-G(M_V/\Cmin)>0$, independently of the candidate pair. Bayes' rule consequently yields a continuous posterior vector $\nu(\mu,L)$ and continuous displayed masses $m_{\mu,L}$.

Map the candidate pair into
\[
 \left(\nu(\mu,L),\ \sum_z m_{\mu,L}(z)A(\nu(\mu,L)(z))\right).
\]
The second component lies in $[A_{\min},A_{\max}]$, since masses sum to one. Brouwer's theorem gives a fixed point. At it, the normalizer is the denominator in \eqref{eq:allocation}, beliefs are Bayesian, and seller choices are optimal. The budget equality is immediate by summing the allocation.

At a fixed point, $L$ is common to all choices, and $N A(t)/L+\rho D(t)$ is strictly increasing in $t$. If an adjacent higher rating were no more credible, it would be strictly dominated for every entrant by the cheaper lower rating. The composition argument used in Theorem~\ref{thm:location} again gives a contradiction. Thus every equilibrium has strictly increasing credibility. Nothing in this argument compares welfare across different allocation formulas.

\clearpage
\section{A review technology for the target cost}\label{app:technology}
This appendix connects the target cost used in the model to the reviews needed to reach a displayed rating. A seller begins with authentic review mass $n_0>0$ and exact average $x_0$. Reviews are modeled as divisible mass: the seller can add mass $y_s\geq0$ of fake reviews at each score in a finite set $\mathcal S$, whose highest score is $\bar s$. The new average combines authentic and fabricated reviews:
\[
 x(y)=\frac{n_0x_0+\sum_{s\in\mathcal S}s y_s}{n_0+\sum_{s\in\mathcal S}y_s}.
\]
Displayed rating $z$ corresponds to a bin that includes its lower edge $z^-<\bar s$. For an upward target with $x_0<z^-$, reaching that bin requires
\begin{equation}\label{eq:reviewconstraint}
 \sum_s(s-z^-)y_s\geq n_0(z^--x_0)=:L>0.
\end{equation}
The right side of \eqref{eq:reviewconstraint} is the total lift needed to bring the authentic reviews up to the target average. Each fake review with score $s$ contributes $s-z^-$ toward that lift. A review below the target contributes negatively, and one exactly at the target contributes nothing. This accounting explains why only reviews above the target will be useful.

Suppose producing review mass $y_s$ at score $s$ has base cost $c_s(y_s)$, where $c_s$ is continuously differentiable, strictly increasing, and strictly convex, with $c_s(0)=0$. A seller with private cost parameter $\kappa$ pays $\kappa\sum_s c_s(y_s)$. Define $C(x_0,z)$ as the minimum of $\sum_s c_s(y_s)$ subject to \eqref{eq:reviewconstraint}. The next lemma shows that this minimum provides the increasing target-cost schedule assumed in the main text.

\begin{lemma}[Least-cost review production]\label{lem:technology}
For $x_0<z^-$, the minimizer exists, is unique, uses only scores above $z^-$, and satisfies \eqref{eq:reviewconstraint} with equality. There is a unique $\lambda>0$ such that
\[
 \frac{c_s'(y_s)}{s-z^-}=\lambda\quad\text{for active scores},\qquad
 c_s'(0)\geq\lambda(s-z^-)\quad\text{for inactive scores above }z^-.
\]
For a fixed target, the minimum cost strictly decreases with $x_0$ below the target edge. For fixed $x_0$, it strictly increases as the target edge rises above $x_0$.
\end{lemma}

For an active score, $c_s'(y_s)/(s-z^-)$ is the marginal cost per unit of lift toward the target. At the optimum, this cost is the same across all scores in use; otherwise reviews could be shifted toward a cheaper way to raise the average. An unused score has no lower marginal cost per unit of lift. The multiplier $\lambda$ is the marginal cost of requiring one additional unit of lift.

\begin{proof}
We first remove reviews that do not help reach the target, then establish the least-cost allocation among those that do. For any feasible profile, let $P=\{s:s>z^-\}$ and $Q=\sum_{s\in P}(s-z^-)y_s\geq L$. Replace the profile by $\widehat y_s=(L/Q)y_s$ for $s\in P$ and zero elsewhere. This preserves feasibility and weakly reduces each coordinate. Unless the original profile already uses only $P$ and meets the constraint exactly, a coordinate falls strictly, so strict increase of costs gives a strict cost reduction. Because $z^-<\bar s$, the set $P$ is nonempty and the problem is feasible. Every coefficient $s-z^-$ in the reduced problem is strictly positive. The problem therefore reduces to the compact simplex
\[
 \{y_P\geq0:\sum_{s\in P}(s-z^-)y_s=L\}.
\]
Continuity gives existence and strict convexity gives uniqueness. The first-order conditions are the stated multiplier equalities and inequalities. For the equivalent reduced inequality problem, all useful coordinates can be chosen positive and sufficiently many highest-score reviews make the target constraint strict. Slater's condition therefore makes the multiplier conditions necessary. They are sufficient by convexity. At least one coordinate is positive. Strict convexity and strict increase imply $c_s'(y_s)>0$ at every positive $y_s$, and an active coordinate uniquely determines $\lambda>0$.

If $x_0$ increases while remaining below the target, the required lift falls. Scaling the old nonzero optimum by the ratio of the new lift to the old lift gives a feasible strictly cheaper profile. Hence cost strictly decreases. If the target edge rises from $z^-$ to $z'^{-}$, an optimum for the higher target satisfies
\[
 \sum_s(s-z^-)y_s
 =n_0(z'^{-}-x_0)+(z'^{-}-z^-)\sum_s y_s
 >n_0(z^--x_0).
\]
It thus exceeds the lower target and can be strictly reduced in cost by the first argument. The higher target has strictly greater minimum cost.
\end{proof}

To recover the main model, hold authentic review mass $n_0$ and the review-production functions $c_s$ fixed across honest origins, and associate each displayed honest rating $a$ with an exact average $x_0(a)$ inside its bin, increasing across bins. Set $C(a,a)=0$ and $C(a,z)=C(x_0(a),z)$ for $z>a$. Higher targets then cost strictly more, as required by Assumption~\ref{ass:cost}(i). Conversely, a better starting average lowers the cost of reaching the top, so $C(a,z_K)$ is strictly decreasing in $a<z_K$. This gives the cost ordering used in Corollary~\ref{cor:rate}. Taking the same authentic review mass $n_0$, representative exact average, and review-production costs for both qualities within a category ensures that the technology is quality-neutral conditional on the honest rating. These are sufficient restrictions for the microfoundation; a rounded score alone need not identify an exact average or a review count in actual data.

If all scores have the same differentiable strictly convex cost function, the multiplier conditions imply that every score above an active score is also used and receives more reviews. Higher scores provide more lift, but rising marginal production costs can make a mixture of scores economical. The numerical illustration uses a simpler technology: only the highest score $\bar s$ is available, with a constant unit placement cost. The required mass of fake reviews, and hence the base cost, is then
\[
 C(a,z)=\frac{n_0[z^--x_0(a)]}{\bar s-z^-},\qquad z>a.
\]
With $n_0=1$, $x_0(a)=a$, $z^-=z-0.05$, and $\bar s=5$, this is \eqref{eq:numericalcost}. Strict convexity is useful for a unique multi-score profile but is not required for the main equilibrium results.

Within the model's upward action set, the seller can therefore solve its problem in two stages without losing any feasible optimum. It first finds the cheapest review profile for each displayed target $z\geq a$, then chooses among those targets using \eqref{eq:payoff}. A more expensive profile with the same displayed rating yields the same return at greater cost and is dominated. For an upward target, the least-cost profile reaches the lower edge of its bin exactly, so it produces the intended displayed rating. The target-cost formulation is therefore equivalent to choosing a complete review profile subject to the same upward restriction.

\section*{Data and code availability}
All numerical illustrations are generated from parameters reported in this paper. The accompanying source bundle contains the replication program, generated equilibrium data, and a numerical verification report.

\end{document}